\documentclass{article}
\usepackage[utf8]{inputenc}
\usepackage[T1]{fontenc}

\usepackage{amsmath, amsthm, amssymb}
\newcommand{\eps}{\varepsilon}
\usepackage[margin=1.3in]{geometry}

\usepackage{xcolor}

\usepackage[colorlinks=true, allcolors=blue]{hyperref}
\usepackage[nameinlink,capitalise]{cleveref}
\hypersetup{
    citecolor={violet}
}

\usepackage{xspace}

\theoremstyle{plain}
\newtheorem{theorem}{Theorem}[section]
\newtheorem{observation}[theorem]{Observation}
\newtheorem{lemma}[theorem]{Lemma}
\newtheorem{corollary}[theorem]{Corollary}

\newtheorem{definition}[theorem]{Definition}
\newtheorem{question}[theorem]{Question}

\theoremstyle{definition}
\newtheorem{remark}[theorem]{Remark}

\newcommand{\omri}[1]{\textcolor{purple}{Omri: #1}}
\newcommand{\madhu}[1]{\textcolor{blue}{Madhu: #1}}
\newcommand{\elchanan}[1]{\textcolor{red}{Elchanan: #1}}

\newcommand{\none}{{\xspace\bot\xspace}}
\newcommand{\COUNT}{{\textsc{count}\xspace}}
\newcommand{\origin}{{\textsc{origin}\xspace}}

\newcommand{\N}{\mathbb{N}}
\newcommand{\E}{\mathbb{E}}
\newcommand{\cP}{\mathcal{P}}

\newcommand{\run}{\mathrm{run}}

\newcommand{\Bin}{\text{Bin}}
\newcommand{\Geom}{\text{Geom}}

\newcommand{\barg}{\overline{g}}
\newcommand{\barZ}{\overline{Z}}
\newcommand{\barCOUNT}{\overline{\COUNT}} 
\newcommand{\barorigin}{\overline{\origin}} 
\newcommand{\dcpl}{(\barg,\barZ,\barCOUNT,\barorigin)}

\title{Information Spread with Error Correction}

\author{Omri Ben-Eliezer\thanks{Center of Mathematical Sciences and Applications,  Harvard University. Email: \texttt{omribene@cmsa.fas.harvard.edu}} \and Elchanan Mossel\thanks{Massachusetts Institute of Technology. Email: \texttt{elmos@mit.edu} . Supported in part by a Simons 
Investigator Award, Vannevar Bush Faculty Fellowship ONR-N00014-20-1-2826 and NSF awards 
DMS-2031883 and CCF 1918421} 
\and Madhu Sudan\thanks{School of Engineering and Applied Sciences, Harvard University, Cambridge, Massachusetts, USA. Supported in part by a Simons 
Investigator Award. Email: \texttt{madhu@cs.harvard.edu}.}%
}
\date{}

\begin{document}

\maketitle

%\madhu{Don't let the label info in the pdf scare you. If it is distracting, you can comment out two lines near the top of the pdf}
\begin{abstract}
%We study the process of information dispersal in a network with communication {\em errors}, considering a simple model where a single bit of information initially known to a single source is dispersed through the network, and communication errors lead to differences in the agents' opinions on this information. 
We study the process of information dispersal in a network with communication {\em errors} and local {\em error-correction}. Specifically we  consider a simple model where a single bit of information initially known to a single source is dispersed through the network, and communication errors lead to differences in the agents' opinions on this information. 
%Naturally, such errors can very quickly make the communication completely unreliable, and in this work we study to what extent these errors can be mitigated by local {\em error-correction} where nodes periodically check their knowledge against that of (some subset of) their neighbors. 
Naturally, such errors can very quickly make the communication completely unreliable, and in this work we study to what extent this unreliability can be mitigated by local {\em error-correction} where nodes periodically correct their opinion based on the opinion of (some subset of) their neighbors. 
We analyze how the error spreads in the ``early stages'' of information dispersal by monitoring the average opinion, i.e., the fraction of agents that have the correct information among all nodes that hold an opinion at a given time. Our main results show that even with significant effort in error-correction, tiny amounts of noise can lead the average opinion to be nearly uncorrelated with the truth in early stages.
%o almost half the agents with opinions holding incorrect opinions in early stages. 
We also propose some local methods to help agents gauge when the information they have has stabilized. 
%\omri{Made minor edits. Feel free to comment or revise. Not sure if the name ``average opinion'' is good, especially that it is not used in the main body.}
%\madhu{Made further edits. I like the abstract and the term "average opinion" which is something we do care about even if we don't use the phrase. My changes can be detected by looking at the commented out tex. Mainly - the process we study has errors + error-correction built in - so the first change. Second change mainly to spell out what correction is.}
\end{abstract}

\section{Introduction}

Societal knowledge is acquired by a collection of complex intertwined distributed processes involving agents that learn background facts, gather new data, make inferences, and communicate this to other agents. Each one of these steps is susceptible to noise --- so much so that in the absence of error-correction mechanisms much of this knowledge could be totally flawed. Of course some natural error-testing and correcting mechanisms are built into this network of complex processes and this brings in large amounts of reliability to existing knowledge. But much of the design and adoption of the mechanisms is based, at best, on empirical analysis and very little theory captures these processes.

With this broader context in mind, in this paper we explore one corner of this space: The simple spread of information in a society. 
%(So we omit the more complex challenges associated with building inferences on top of potentially incorrect knowledge.) 
%\elchanan{Removed: (So we omit the more complex challenges associated with building inferences on top of potentially incorrect knowledge.) as there are many other things we ignore ...} 
We consider benign errors in this setting where errors creep in when agents communicate a bit of information along. We consider a natural error correction processes where agents simply check their information against that of their neighbors and (with some probability) revise their information. This simple model allows us to ask questions of the form: How much error-correction ``effort'' is necessary/sufficient to ``protect'' from a certain amount of error? How does the structure of the network influence the prevalence of error? How can individual agents aim to gain confidence in the correctness of the information they possess, at any given moment of time?

We turn to our specific model next but as a teaser we describe some of the qualitative results we obtain in our model. We find that expansion in networks can be harmful to correctness: Specifically in expanding networks while information is still spreading, most people that hold any information actually hold information uncorrelated with the truth! And this holds even when the fraction of errors is tiny and the effort in checking/correcting is overwhelming. But on the positive side, the timing information that agents hold --- ``when did they first hear about this information?'' and ``when did they last change their opinion'' --- can be very useful in helping the agents assessing the quality of their information. With these potential results in mind we turn to our specific model.

\subsection{Our Model}\label{ssec:intro-model}

We describe our model formally in Section~\ref{sec:model}. Here we introduce the model less formally but explain some of the choices. Two key principles we follow are: (1) Keep the model simple focusing on the errors and correction; and 
%(2) Keep the model ``cooperative'': That is, the agents' goal is to cooperate to learn the truth quickly --- not to weigh in with their own opinions on what is true. (I.e., the agents are myopic.) 
(2) Avoid strategic considerations by assuming agents believe the information they have at each point and transmit it to their neighbors. (I.e., the agents are myopic.) 
%\elchanan{This is not really a cooperative scenario. In a cooperative scenario, agents may withhold information / not broadcast for the benefit of others. E.g. The can transmit only when they are sure.
%I would change the sentence (2) to: avoid strategic considerations, assuming agents believe the information they have at each point and transmit it to their neighbors}
%\madhu{Sounds good. IMplemented the change. Old tex commented out in latex.}

These choices may make the model more ``optimistic'' than some natural settings, but note that this makes the negative results even more significant.

Turning to our model, we consider a locally finite graph $G$ ($G$ could be finite or countably infinite) where a root vertex $r$ is given a bit of ``true'' information $f_0(r) \in \{0,1\}$, and the eventual goal is to disseminate this bit to the rest of the graph. All other nodes know they don't know this bit initially, indicated by setting $f_0(v) = \none$ for every $v \ne r$. We refer to such a node as one with no ``opinion'' on the truth. 

%\madhu{changes here}
While our model allows the communication graph to be arbitrary, actual communication only happens in a {\em directed spanning tree} $T$ whose edge directions are from the root outward. (We note that keeping the communication graph acyclic is crucial to our analysis --- we elaborate further on this choice below.) There are two natural choices for selecting this tree $T$. The tree could be somehow fixed in advance, independent of the actual information spread process. We refer to this as the {\em offline} version, and this is the version we pursue in most of this paper. Another natural choice, which we call the {\em online} version, would be to allow the tree to grow randomly, potentially correlated with the information spread process. We discuss this option briefly in Section~\ref{sec:c_less_1}. 
In our main model, $T$ is a \emph{breadth-first} (BFS) tree\footnote{
A \emph{breadth first (BFS) tree} in a graph $G$ with root $r$ is any tree $T$ satisfying the following: for every non-node root $v$, the distance $d_G(r, v)$ between the root $r$ and $v$ in $G$ is equal to the distance $d_T(r,v)$ between them in $T$.}  given in advance, where the path length from the root to any node $v$ equals to the distance between them in the original graph. (Our analysis can easily extend to the setting of non BFS trees, with the caveat that any quantitive aspects depend on a certain expansion measure of the tree $T$ rather than the graph $G$. These two turn out to coincide when the tree is a BFS tree.) %{\color{purple} FIX THIS. The tree can either be constructed offline -- fixed in advance -- or online, that is, constructed gradually as nodes form their initial opinion. We also briefly mention extensions of the model in which the tree may take a different shape.} \omri{We will take the offline model, online model in Section 9.}
%{\color{purple} Another possible design choice is to construct the tree online, with the edge leading to a node $v$ created immediately before the initial opinion of $v$ is formed. (Say something about relations between offline and online.)}  
%and this tree is a breadth-first search tree \omri{Revise here}.  
Communication happens at discrete time steps along this tree, and in particular at every time step an unopinionated node $v$ checks whether its parent in $T$
has an opinion; if it does, $v$ forms an opinion based on the opinion of the parent.
%with its parent in the tree to see if the parent has an opinion or not, and forms an opinion if the parent has one. 
%\madhu{end changes}

Now we turn to errors. {\em Errors} enter only when a previously unopinionated node is forming an opinion, by talking to its parent in the rooted tree. The probability of this error is one of the basic parameters of our model. The second parameter focuses on the {\em error-correction} where in each step with some probability, an opinionated node checks its own opinion with that of its parent and updates to the parent's opinion if they disagree. The probability of this error-correction (``fact-checking'') operation is the second basic parameter in our model and together these are the main two parameters we consider. %\omri{How should we discuss the third parameter here? earlier in this subsection I mentioned that we consider ``an extension'' to the BFS model.} 
We stress that no errors happen during this error correction phase. Obviously, we could have allowed some errors to creep also during a legitimate update step where a node with an opinion differing from its parent updates to the wrong value, but this kind of error can simply be modelled by a lower probability of an error-correction step. But crucially if a parent and child agree on an opinion, the child does {\em not} update to a disagreeing opinion due to communication error. This choice simplifies the underlying process. In particular once all nodes have an opinion equal to that of the root, then no further updates happen. This is the unique fixed-point in the space of configurations and a goal could be to understand when this point is reached, but most of our analysis aims to understand the opinions before this state is reached. 

As mentioned, when designing the model we aimed for simplicity and fast convergence. Below we describe why the model checks these two boxes.
\begin{itemize}
\item \textit{Simplicity.}
Having directed cycles in the information graph can lead to a persistence of errors. This is the phenomenon that lies at the heart of many studies in ``opinion dynamics'' and forces agents to start having beliefs in the truth.
By picking a model which only communicates information in the same direction -- away from the root -- we shy away from this more-studied phenomenon which would add substantial complexity to our model. We could have picked DAGs rather than trees, or other sophisticated error correction mechanisms, but these again lead to more technical challenges and do not seem to provide a substantially different picture on a conceptual level. %Again, our aim is to get insights by looking at simple processes. 
%
%\madhu{changes here: Added below}
Similarly the choice of working with the offline version is based on the desire for simplicity. In an online version where the tree is allowed to grow randomly, the structure and parameters of the tree itself can be complex to analyze. We elaborate more on the challenges in Section~\ref{sec:c_less_1}. Staying with the offline model takes another potential source of complexity out of our model.
%\madhu{end changes}

\item \textit{Fast convergence.} As we demonstrate (see Section~\ref{sec:slower}), slight modifications in the process -- for example, allowing a node to listen to all of its neighbors in the graph or in the tree during error correction, rather than to a single fixed parent in the tree -- may lead the convergence rate of the process to be exponentially slower, which is arguably less realistic. 
\end{itemize}

\subsection{Our analysis}\label{ssec:intro-analysis}

We describe our results more formally in Section~\ref{sec:results}. Here we give an informal description of the nature of questions we are able ask and answer in our model. 

In our simple model, information propagates in two ``waves''. The first wave is the {\em rumor frontier} where nodes become aware of the existence of a new fact. The second wave is the {\em truth frontier} which is the connected subtree including the root where all nodes agree in opinion with the root. Both frontiers make monotone progress. The former proceeds deterministically (due to our model choice) whereas the latter is probabilistic and the main analytic task is to determine how fast this frontier moves. Furthermore to study this speed it suffices to look at just the directed path and we will focus on this case here (as well as in most of the paper). 

It is easy to make a back of the envelope calculation that takes the two parameters of the model and heuristically estimates the speed of the truth frontier. This would suggest a speed that depends only on the second parameter (error-correction effort) and only needs the other parameter (error probability) to be positive. But these heuristics are just that and a careful examination reveals that in order to convert this to a formal analysis one needs to understand the pattern of opinions in the nodes  between the two frontiers. In particular it is possible for the truth frontier to make massive jumps forward in a single step by swallowing a brief run of incorrect opinions sandwiched between the truth frontier and another long run of opinions which accidentally happens to agree with the truth. (Indeed our analysis shows that such events do happen in almost every run.) 

In order to convert the heuristic analysis to a formal one, we introduce a more detailed model that we call the ``decoupled model'' which tracks not only the current opinion of a node but also where this opinion originated. (The decoupled model is similar to the coalescing random walk representation of the ``voter'' model studied in \cite{HolleyLiggett:75, Liggett:85} and we elaborate on the relationship in Section~\ref{ssec:related}.)  Speeds of the frontier movement in this model are easy to analyze but we still need to analyze the ``runs of opinions'' in this model to get information about our original model. By doing so we roughly observe the following phenomenon: The gap between the truth frontier and the rumor frontier grows linearly with time (with exactly the same constant as predicted by the heuristic calculation!). At time $t$ there are roughly $\sqrt{t}$ alternations of opinions along the path from the truth frontier to the rumor frontier. Furthermore a node at distance $t'$ from the rumor frontier is expected to be in a run of length roughly $\sqrt{t'}$ --- so the runs get shorter (or alternations become more frequent) as one approaches the rumor frontier. 

Converting the information on run lengths and speed of the truth frontier, especially from our decoupled process, immediately translates to information about the opinions of nodes and their correlation with the truth. Roughly each run outside the truth frontier has a random opinion independent of the truth and of other runs. Since in expanding graphs most of the nodes are close to the frontier, most of the nodes thus are part of small sets whose information is independent of that of others. Thus with high probability, there is very little correlation between the truth and the opinions of most nodes.

Finally, note that while nodes inside the rumor frontier are aware of this fact (by virtue of the fact that their opinion is not $\none$) nodes within the truth frontier are not! So how can a node {\em know} if its value is correct? Our analysis suggests a simple method: A node should count time since it formed an opinion, and measure how recently it last changed its opinion. If the former is $t_0$ and the latter is $t_1$ and $t_1 > \tilde\omega(\sqrt{t_0})$, then its opinion is likely to be the truth, and if not the opinion is likely to be {\em uncorrelated} with the truth. Thus there is a {\em sharp transition} in the confidence of a node about its opinion.

We now briefly give pointers to where the results alluded to may be found in the paper. In Section~\ref{sec:model} we introduce our model formally. In Section~\ref{sec:paths_and_trees} we describe the heuristic bound on the speed of the truth frontier and prove it is a lower bound on the actual speed. In Section~\ref{sec:analysis_path} we introduce the ``decoupled model'', show the equivalence of this model to our original one, and start analyzing this model. In Section~\ref{sec:truth_frontier} we formally prove the upper bound on the speed of the truth frontier. In this section we also describe how agents can estimate their confidence in their opinion. All the analysis upto this point is on the path, though it also applies to branches in the more general tree setting. In Section~\ref{sec:unreliability} we relate the results to general graphs and in particular to the fraction of agents with the correct opinion. In Section \ref{sec:slower} we show how slightly different choices of model would lead to much slower convergence time.

\subsection{Related Work}\label{ssec:related}

Without error correction, the model presented here on trees is identical to the well studied {\em broadcast model on the tree}.
For this model, without correction, it is well known that asymptotically half of the nodes will take each of the two values~\cite{KestenStigum:67}. Finer result regarding the lower order fluctuations 
of the fraction of nodes taking each value were studied in~\cite{KestenStigum:66}. 
More recently questions of root reconstruction on trees deals with {\em global} estimators that can estimate the root value better than random, given the information at all nodes sufficiently far 
away~\cite{BlRuZa:95,Ioffe:96a,Ioffe:96b,EvKePeSc:00,BCMR:06}. The phase transitions proven for the existence of such estimators do not play any role in our results. 
%~\cite{BlRuZa:95,Ioffe:96a,Ioffe:96b,EvKePeSc:00,JansonMossel:04,BCMR:06}

Work on noisy-computation considered how can one compute the value of $f(x)$ where $f$ is given by a Boolean circuit, when one can only use  noisy gates~\cite{vonNeumann:56,EvansSchulman:99}. 
The main result of the area is that for small enough error rates this can be achieved for circuit of polynomial size with a logarithmic blowup in the depth.  
The problem is similar to the problem we consider in that the goal is to have no errors in computing $f$.  
 However, the error correction performed is stronger as one uses (the majority of) a number of bits at each gate. 
A related broadcast model where 
nodes compute their values by computing the majority of a number of nodes in previous generations was studied in~\cite{MaMoPo:20}. 

There is a huge body of work studying opinion exchange dynamics and learning on networks. Some of the main models in the area include the Voter model, the DeGroot model, SIR, SIS and other epidemics models, and Bayesian and non-Bayesian learning models, 
see e.g.~\cite{Jackson:08,MontanariSaberi:10,MosselTamuz:17}.   

The epidemics models, the Voter and the DeGroot model all assume that each agent has a state/belief at the beginning of the process and opinion exchange results in agents updating their opinion iteratively. This is different than the model studied here, where there is a new topic, where initially most agents have no opinion. Furthermore, the basic interaction in most such models is supposed to model persuasions, and infections while our model just focuses on information sharing and error correction. 

Among all the models in the preceding paragraph, the {\em voter model}~\cite{HolleyLiggett:75} is the closest to the model we study in this paper. In fact, if we consider a variant of our model where error correction is performed with respect to a randomly chosen neighbor having an opinion, then after the initial broadcast process, and excluding the root, the model behaves very much like the voter model. One consequence of this fact is that the convergence time for simple graphs such as the binary tree, is exponential in the radius of the graph, versus linear in our model, see Section~\ref{sec:slower}. 

We note further that our decoupled model and the coalescing random walk representation of voter model are quite similar as representing they track the information backward using random walks, see e.g.~\cite{Liggett:85}.

Some of the work on Bayesian learning with myopic agents in economics is also concentrated on information aspects on networks. However, much of this work requires nodes to be computationally unbounded, and has no realistic bounds on the time of convergence of the learning 
process~\cite{HJMR:19b}

The more computationally efficient models in the literature, those studying {\em herding}, allow nodes to take action only once, and thus have no error-correction. Many of them also require each node to observe all nodes preceding it in a fixed order~\cite{BiHiWe:98,Banerjee:92,SmithSorensen:00}. 

There is also a huge body of work on misinformation on social networks, including popular 
books, see for example~\cite{Wylie:19,ScheufeleKrause:19}. 
These works study topics such as large scale targeting, identifying influencers,
assimilation bias, the effect of surprise in spread and others. 
To the best of our knowledge none of this work studies the effect of neutral communication errors and error correction on misinformation.

\section{Model Definition}
\label{sec:model}

%\madhu{Revised plan for this section. Discuss a model in terms of three parameters. Discuss implications of setting third parameter to 1. Define 2 parameter model and say we will only present results about this model. (So only one formal definition which is the stuff below without red stuff.) Say what results say about the three parameter model. }
%\omri{Done -- except for the last sentence that we need to finalize.}

%\madhu{changes here}
We now turn to the formal definition of the model. Recall that the model is specified by three structural parameters $G$, $r$ and $T$ and two numerical parameters $a$ and $b$. 

\begin{definition}[Information spread with error correction]
\label{def:model_general_Graphs}
Fix two parameters $a,b \in [0,1]$.
Let $G = (V,E)$ be an undirected graph, let $r \in V$ denote some \emph{root} vertex, and let $T$ be any BFS-tree of $G$ rooted at $r$. Consider the following process, proceeding in rounds. 

\begin{itemize}
\item At any given round $t \geq 0$, each vertex $v \in V$ holds a label $f_t(v)$, which is one of three values: $1$ (holds ``correct'' opinion), $-1$ (``incorrect'' opinion), or $\none$ (no opinion).
\item Initially, only the root holds an opinion (the correct one): $f_0(r) = 1$, and $f_0(v) = \none$ for $v \neq r$. We also define a \emph{parent} function $p \colon V \to V$, where $p(v)$ is the unique neighbor of $v$ that lies on the path from $r$ to $v$ in the tree $T$, and $p(r) = r$.
\item Given the values of $f_{t-1}(v)$ for all $v \in V$, the update rule defining $f_{t}(\cdot)$ is as follows. 
\begin{itemize}
    \item \textbf{Before opinion formation.} If $f_{t-1}(v) = \none$ and the set $N_{t-1}(v) := \{u \in V: (u, v) \in E, f_{t-1}(u) \neq \none \}$ is empty, then $f_t(v) = \none$.
    \item \textbf{Initial opinion.} If $f_{t-1}(v) = \none$ and $N_{t-1}(v)$ is non-empty, we %pick $p(v)$ uniformly at random from $N_{t-1}(v)$, and 
    set
     \[
        f_t(v) = \begin{cases}
        %{\color{red} \bot }&{\color{red} \text{w.p. } 1-c, }\\
        %{\color{red} f_{t-1}(p(v)) }&{\color{red} \text{w.p. } c - a, }\\ 
        f_{t-1}(p(v))  &\text{w.p. } 1-a,
        \\
        -f_{t-1}(p(v)) &\text{w.p. } a.
        \end{cases}
     \]
    \item \textbf{Opinion updates.} If $f_{t-1}(v) \neq \none$ then
        \[
        f_t(v) = \begin{cases}
        f_{t-1}(p(v))  &\text{w.p. } b,
        \\
        f_{t-1}(v)) &\text{w.p. } 1-b.
        \end{cases}
        \]
\end{itemize}
\end{itemize}

%\omri{We decided to first state the model in full generality, including the parameter $c$. Immediately after we will say that we consider $c=1$ until further notice, as it is a special case simpler to analyze. In one of the later sections we can discuss the case where $c \neq 1$ in more detail.}

Formally, thus, the {\em Information Spread} model, or simply ``Spread Model'', is given by the random variables $\{f_t(v)\}_{t \in \N, v \in V}$, where $f_t(v) \in \{-1,\bot,1\}$ for every $t,v$ are generated as above. The model is parameterized by $(G, r, T, a, b)$ where $G=(V,E)$ is a  graph on a finite or countably infinite set $V$, $r \in V$ is the root vertex, $T$ is a breadth-first search tree in $G$ rooted at $r$, %\madhu{BFS is not well defined without a root}, 
and $a,b \in [0,1]$ are two real numbers. 
\end{definition}
%\omri{We had this issue with the parametrization where $G$ doesn't matter once $T$ is fixed. Need to decide whether $G$ is removed from the tuple of parameters, and proliferate if needed.}
%\madhu{I added the two red lines above just to give us a parameter to discuss. We are working with $c=1$ in this paper, and discuss $c \leq 1$ in Section \ref{sec:c_less_1}. But the para below holds only for $c=1$.}

We remark again that the process above is the ``offline version'' where a communication tree is given in advance. The model can be adapted to the setting where the tree is chosen in an online fashion --- we discuss this model in Section~\ref{sec:c_less_1}.

Here, ``errors'' (or disagreements with the parent) are induced with probability $a$ when a node first forms an opinion, and correction (or being convinced by the parent) refers to the process where a node may change its opinion to the parent's opinion with some probability $b$ in any given round. For any $b$ bounded away from zero, it is not hard to show (see, e.g., Lemmas \ref{lem:lb_truth_frontier} and \ref{lem:run_branch_path}) that all but an arbitrary small constant fraction of the nodes converge to the correct opinion after a number of rounds linear in the diameter of the graph. 

\begin{remark}[Fast convergence and alternative models]
The model we propose is a simple example where opinions converge quickly and with probability one to a consensus -- in a number of rounds at most linear in the diameter of the graph -- despite having errors (or disagreements) in the communication process. Such fast convergence is desirable as it arguably better describes opinion formation and distribution in real world graphs, compared to processes where the convergence rate is polynomial in the graph size. 
%In particular, see Section~\ref{sec:slower} which shows that if a random neighbor is chosen each time for error correction, than for simple graphs such as the complete graph or the binary tree we get an exponentially slower process. 
%(which for many interesting graph classes is exponential in the depth).

We note that similar qualitative results hold for various extensions of the above model as well as other fast converging models with errors; we describe a couple of these in more detail.
\begin{itemize}
    \item In our model, each node $v$ is assigned a fixed parent $p(v)$ from which it subsequently receives all incoming information. 
    %As is shown in Section TODO, listening to all of its neighbors instead of a single designated parent may dramatically increase the time required for convergence. 
    Another possible model with fast convergence properties would be for $v$ to continuously listen to all its neighbors that formed an initial opinion before $v$ itself; in our case, these are precisely all neighbors of $v$ whose distance to $r$ is smaller than $v$'s distance. We note that the main phenomena arising in this paper, including the fact that rumors spread much faster than reliable information, as well as their proofs, all extend to this variant as long as appropriate conditions on the volume growth of the graph hold.
    \item Our model is synchronous, in that it proceeds in discrete rounds; however, essentially the same behavior can be observed in analogous asynchronous models.
    \iffalse
    \item At a first glance, our model does not seem to account for errors in the error-correction mechanism. %or when initially picking a parent. 
    However, we note that errors during error correction (which would cause a node to listen to its parent, but then flip its opinion) can merely be addressed by modifying the parameter $b$. %For the latter, one can think, for example, of a model where a node $v$ has non-zero probability not to form an initial opinion during round $t$, even if $N_{t-1}(v)$ is non-empty. Again, such a modified model would still qualitatively exhibit the same phenomena we observe in this paper.
    \madhu{This is misleading since the model of "errors during error-correction" could introduce errors even when all agents are in agreement. I suggest dropping this item. (I tried to state it clearly but keeping it unambiguous is not easy.)}
    \fi
\end{itemize}
\end{remark}

\begin{remark}[A slower variant]
One natural variant that leads to exponentially slower convergence, is a model where at each round error correction is performed with respect to a randomly chosen opinionated neighbor (instead of a fixed parent, as in our model). 
In Section~\ref{sec:slower} we show that in this slightly modified model, for simple graphs such as the complete graph on $n$ vertices or the binary tree on $n$ vertices the convergence time of the modified model is $\Theta(n)$, while our process converges in time 
$\Theta(\log n)$. 
\end{remark}

\section{Results}\label{sec:results}
All of the results we present in this section are for the model defined above (Definition \ref{def:model_general_Graphs}) with parameters $0 < a < 1/2$ and $0 < b < 1$, which we perceive as constants (that do not depend on the graph $G$). Perhaps the most interesting regime is when $a$ (the error probability) is a small constant and $b$ (the correction probability) is much larger. In particular, all of our results hold for this regime.
%Recall that the process on any graph is defined by first choosing a BFS tree $T$ and then running the process on the BFS tree.\footnote{For a given graph $G$ and root $r$, we note that most results do not rely heavily on the particular choice of $T$ among all choices for possible BFS trees of $G$ rooted at $r$. We also stress again that all results hold equivalently for the online variant of the model discussed in Section \ref{sec:model}.}

Our analysis begins with the simplest tree, that is the (directed) path $\cP_n$ with $n$ being finite or $\infty$, where we prove that the truth frontier moves at speed $b+o(1)$. Recall that the truth frontier (see Definition \ref{def:rumor_truth_frontier}) refers here to the furthest point from the root whose opinion is (i) correct, and (ii) guaranteed to be correct for the rest of the process.

\begin{theorem}[Location of the Truth Frontier in the path; see 
Lemma~\ref{lem:lb_truth_frontier} and Theorem \ref{thm:ub_truth_frontier}]
Consider our process with fixed parameters $0 < a < 1/2$ and $0 < b < 1$ on the (directed) path $\cP_n$.
The distance of the truth frontier from the root after round $t$ is $\min\{(b \pm o(1))t, n\}$ with probability $1-e^{-t^{\Theta(1)}}$. The $o(1)$ term tends to zero at a polynomial rate as $t \to \infty$.
\end{theorem}
The lower bound is easy; the upper bound requires a much more intricate analysis of the typical behavior of nodes ``just beyond'' the truth frontier. In general graphs, the behavior of our process along any particular branch is equivalent to a process with the same parameters over the path, see Lemma \ref{lem:run_branch_path}. Therefore, as a corollary of the above theorem, we obtain tight bounds on the typical progress of the truth frontier along any fixed branch of the tree, see Corollary \ref{coro:truth_frontier}.

We then proceed with a much finer analysis of runs of consecutive values ($+1$ or $-1$) between the truth frontier and the rumor frontier in paths. This analysis allows us to show that nodes between the two frontiers change their opinion quite often.
\begin{theorem}[length of runs; see Theorem \ref{thm:changes_of_opinion_individual}]
Consider our model with fixed parameters $0 < a < 1/2$ and $0 < b < 1$ over a graph $G$ with root $r$ and let $v \neq r$ be a vertex in the graph. For any $\eps > 0$, there exists a constant $C = C(a,b,\eps)$, which does not depend on $G, r, v$, so that with probability at least $1-\eps$, the following holds. In the first $q$ rounds after $v$ forms an initial (non-$\none$) opinion, if the truth frontier has not reached $v$ during these $q$ rounds, then the longest consecutive streak of rounds in which $v$ does not change its opinion is of length at most 
$$C_{a,b,\eps} \cdot q^{1/2} \cdot (\log q)^{3/2}.$$
\end{theorem}

In fact, the proof of Theorem~\ref{thm:changes_of_opinion_individual} allows nodes who keep a record of all opinions that they held, to estimate quickly and accurately if they have converged to the correct value without knowing anything about the graph or the opinion of other nodes: 

\begin{corollary}[When to stop; see Corollary \ref{coro:when_to_stop}]
Consider a one-player game in the same setting as in Theorem \ref{thm:changes_of_opinion_individual}, where an agent is located at vertex $v$, knows $a,b$ in advance, and can only observe the opinions of $v$ throughout the process;
The agent does not have any information about the graph $G$ or the root $r$. The agent has the ability to end the game after any round of the process. Its goal is to end the game as soon as possible \emph{after} the truth frontier has reached $v$. For any $\eps > 0$, there exists a strategy for the agent with success probability $1-\eps$ to end the game no more than $C_{a,b,\eps} \sqrt{t} (\log t)^{3/2}$ rounds after the truth frontier has reached $v$, where $t$ is the total number of rounds that $v$ has been holding a non-$\none$ opinion so far.
\end{corollary}

We next turn to studying the bias of the process. One of the main motivations for studying this model was trying to understand what fraction of the informed nodes hold the correct opinion. 
The \emph{opinion bias} (Definition \ref{def:opinion_bias}) measures
this fraction, where a bias of $1$ indicates that all nodes are correct, and a bias of $0$ corresponds to opinions being \emph{uncorrelated} with the truth.
We first show that in expectation, {\em graph expansion} leads to the latter situation, where opinions are uncorrelated with the truth. 
We let $V_G(x,t)$ denote the volume (number of vertices) of the ball of radius $t$ around $x$. The following result bounds the expected bias in terms of the volume growth:

\begin{theorem}[Expected bias towards correct opinion; see Theorem \ref{thm:unreliability_expectation}]
Consider our model with fixed parameters $0 < a < 1/2$ and $0 < b < 1$ over a graph $G$ with root $r$. The expected bias of the process at time $t$ satisfies
\begin{equation*}
\left(1-e^{-t^{\Theta(1)}}\right) \cdot \frac{ V_G\left(r, bt - \Theta\left(t^{1-C}\right)\right)}{V_G(r, t)}
\leq
\mathbb{E}[\lambda_{G, r}(t)]
\leq 
\frac{ V_G\left(r, bt + \Theta\left(t^{1-C}\right)\right)}{V_G(r, t)} + e^{-t^{\Theta(1)}}
\end{equation*}
where $C > 0$ is an absolute constant. \end{theorem}

Among other examples, Theorem~\ref{thm:unreliability_expectation} implies that
\begin{itemize}
    \item For the path rooted at $0$, the expected bias is $b+o(1)$.
    \item Fix $d > 0$. For the grid graph ${\mathbb Z}^d$, in $d$ dimensions rooted at $0$, the expected bias is $b^d + o(1)$.
    \item For the infinite binary tree, the expected bias is $o(1)$.
\end{itemize}

The following two theorems show that in all of the examples above as well as in other examples, the opinion bias is concentrated around its expected value. 
Interestingly, we have different statements and proofs in the polynomial and super-polynomial growth cases:

\begin{theorem}[Concentration for super-polynomial volume growth; see Theorem \ref{thm:concentration_super_poly}]

Consider our model over a graph $G = (V, E)$ with root $r$ and fixed parameters $0 < a < 1/2$, $0 < b < 1$, and let $T$ be the communication tree of the process, rooted at $r$. Let $t \in \N$ and suppose  
that the following two conditions hold for constants $0 < c \leq C$ and $0 < c' < 1-b$ and a function $\omega \colon \N \to \N$.
\begin{itemize}
    \item $\frac{\omega(x \cdot (b+c'))}{\omega(x)}$ converges to zero as $x \to \infty$.
    
    \item $V_G(v, x) \leq C \cdot x^{\omega(x)}$ for every $v \in B_{t,G}(r)$ and $x \leq t - d_G(r, v)$.
    \item $V_G(r, x) \geq c \cdot x^{\omega(x)}$ for every $x \leq t$. 
\end{itemize}
Then the opinion bias $\lambda_{G,r}(t)$ is $o(1)$ with probability $1-o(1)$, where the $o(1)$ terms tend to zero as $t \to \infty$. Furthermore, if we replace the $C \cdot x^{\omega(x)}$ and $c \cdot x^{\omega(x)}$ upper and lower bounds with expressions of the form $C e^{x^{\Omega(1)}}$ and $c e^{x^{\Omega(1)}}$ respectively (and ignore the first bullet), then the $o(1)$ terms are of the form $e^{-t^{\Theta(1)}}$. 
\end{theorem}

Examples of graphs with super-polynomial volume growth include the aforementioned binary tree, as well as various small-world models aiming to capture real world social interactions.

We note that for the phrasing of the theorem to make sense, the graph $G$ has to be infinite: the third bullet requires that $V_{G}(r,t) \to \infty$ as $t \to \infty$. For example, one can think of a small world model that expands with time $t$ (and satisfies the well behaved volume growth conditions), and argue that as $t \to \infty$, the opinion bias at time $t$ converges to zero. %\omri{Maybe need to define opinion bias somewhere.}

Real world graphs are finite, however; what would be the interpretation in this case?
Suppose that $G$ is a finite graph that (intuitively) exhibits super polynomial growth up to some radius $t_0$ around the root $r$, where we think of $t_0$ as large compared to the (fixed) parameters $a,b$. 
The theorem implies that in most rounds up to time $t_0$, the opinion bias will typically be close to zero.

The next theorem concerns graphs with polynomial growth, with the prime example being the $d$-dimensional grid (for any fixed dimensionality $d$).
\begin{theorem}[Concentration for polynomial volume growth; see Theorem \ref{thm:concentration_poly}]
Consider our model with parameters $a \in (0,1/2)$ and $b \in (0,1)$ over a graph $G$ with root $r$ at some time $t > 0$. Let $c > 0$ be an arbitrary small constant. Suppose that for every $s < s' \leq t$ with $s' - s \leq t^{0.5 + c}$, it holds that $V_G(r, s') - V_G(r,s) = o( V_G(r,t))$, where the $o(\cdot)$ term tends to zero as $t \to \infty$. Then the opinion bias at time $t$, $\lambda_{G,r}(t)$, is within $\pm o(1)$ of its expectation with probability $1-o(1)$.
\end{theorem}

\paragraph{Other biases}
The results above characterize scenarios where the opinion bias is close to $0$ 
or some other fixed value. 
We note that for finite graphs, growth of balls always stops. 
Thus for finite graph Theorem~\ref{thm:unreliability_expectation} provides only partial information about the bias of the process. 

For example, hypothetically it is possible, even in a graph with an initial exponential volume growth around the root, that at the time when the {\em last node} 
in the graph receives an initial piece of information, almost all nodes already have the correct opinion. 

To understand why this might be true, consider a binary tree on $n$ nodes rooted at $r$ and fix $v$ to be one of the nodes of the leaves. Attach a path of length $g(n)$ to $v$ and call the resulting tree $T$, where $g(n) \gg \log n$. Note that at the time when the last node obtains information in this process, all the nodes in the original binary tree have the correct information. Thus if 
$g(n) \ll n$, at this time the bias is $1-o(1)$. Similarly if we choose 
$g(n) = c n$, then by the time $n + \log^2(n)$ nodes receive information, almost all of them will have the correct information. 

It is not clear if such a pathological tree can be the BFS tree of finite bounded degree expanders.

\section{Paths and Trees: Model and Basic Observations}
\label{sec:paths_and_trees}

As discussed in Section~\ref{sec:model}, our model for general graphs $G$ (with root $r$) can be reduced to an information spread model over a \emph{directed tree}, where the tree is picked among all BFS trees of $G$ rooted at $r$. 
We now shift our focus toward these directed rooted trees. Consider an arbitrary fixed rooted tree $T$ whose edges are directed from the root $r$ outward. Thus, each node $v \neq r$ has a unique parent $p(v)$. 
For convenience, we now redefine (and later proceed to analyze) the part in Definition \ref{def:model_general_Graphs} concerning trees. 

%\begin{definition}[Parent in rooted tree]
%\label{def:tree_parent}
%Let $T = (V,E)$ be a tree and fix a root $r \in V$. The \emph{parent} of a node $v \neq r$ in $T$, denoted $p(v)$, is the neighbor $u$ of $v$ on the (unique) shortest path between $r$ and $v$. For completeness, we define $p(r) = r$.
%\end{definition}

%\madhu{Parent is already defined in the process (def 2.1)}

\begin{definition}[Information spread on rooted tree]
\label{def:rooted_tree}
Fix parameters $a, b \in [0,1]$ and a tree $T$ with root $r$.
Each node $v$ has a label $f_t(v)$ at any time $t$, which is one of three values: $+1$ (holds "correct" opinion), $-1$ (holds ``incorrect'' opinion), or $\none$ (no opinion). The root is always marked by $1$, that is, $f_t(r)=1$ for all $t$, and any other node $v$ is initially marked by $\none$, that is, $f_0(v) = \none$.

The process proceeds in rounds, where in each round, each node $v \neq r$ may listen to its parent in the tree $p(v)$ and change its opinion accordingly. 
The transition rule for round $t > 0$ is as follows.
If $v$ does not have an opinion but its parent does, that is, $f_{t-1}(v) = \none \neq f_{t-1}(p(v))$, then 
\[
        f_t(v) = \begin{cases}
        f_{t-1}(p(v))  &\text{w.p. } 1-a,
        \\
        -f_{t-1}(p(v)) &\text{w.p. } a.
        \end{cases}
\]
If $v$ already has an opinion, that is, $f_{t-1}(v) \neq \none$, we set
 \[
        f_t(v) = \begin{cases}
        f_{t-1}(p(v))  &\text{w.p. } b,
        \\
        f_{t-1}(v)) &\text{w.p. } 1-b.
        \end{cases}
\]
Otherwise (if $f_{t-1}(v) = f_{t-1}(p(v)) = \none$), we have $f_t(v) = \none$.
\end{definition}
Note that it is impossible for a node with an opinion to return to a ``no opinion'' state, but that a node can generally flip its opinion multiple times.
%Intuitively, one can think about the model as inducing an error (disagreement with parent) with probability $a$ when a node first forms an opinion, whereas in subsequent rounds, error correction is carried with probability $b$. Here, error correction means that the node takes the parent's opinion. We note that such a model can also tolerate errors during the error correction phase (adding these is equivalent to considering the original model with a modified $b$). Another possible extension is to allow a node $v$, during its initial opinion formation step, to remain with no opinion with some non-zero probability $c$. Such an extension would not materially change the rest of the discussion unless $c$ is very large (specifically, $c \geq 1-b$), an arguably unrealistic scenario where the probability of a node to change its opinion is larger than the probability to form an initial opinion.
%Given the above model
Now, one may ask what is the typical dynamic of opinions in such a model. First, 
does the process eventually converge to the correct opinion? for example, given any fixed $v$, does the probability that $f_t(v) = 1$ tend to one as $t \to \infty$?
The answer to the latter question is positive. This follows using the next observations, whose proof is immediate from the definition of the process.
\begin{observation}

Let $r = v_0 \to v_1 \to \ldots \to v_i = v$ be the unique path from the root to a node $v$, and suppose that $f_t(v_j) = 1$ for all $0 \leq j \leq i$. Then $f_{t'}(v_j) = 1$ for all $t' \geq t$. 
\end{observation}

\begin{observation}
Let $r = v_0 \to v_1 \to \ldots \to v_i = v$ be the unique path from the root to a node $v$, and suppose that $f_t(v_j) \in \{-1,1\}$ for all $0 \leq j \leq i$. Then $f_{t'}(v_j) \in \{-1,1\}$ for all $t' \geq t$. 
\end{observation}
%For now and until the end of Section \ref{sec:analysis_path} we focus on the case where the tree is the path $$r = v_0 \to v_1 \to v_2 \to \cdots \to v_n\;;$$
%we note that the definitions for directed trees are analogous, and present the generalized analysis in Section \ref{sec:analysis_trees}.
Equipped with the above observations, we can define two types of frontiers: the \emph{rumor frontier}, which marks the collection of furthest nodes from the root that currently hold a (non-$\none$) opinion; and the \emph{truth frontier}, which marks the collection of furthest nodes that hold the correct opinion and will continue to hold it for the rest of the process. Formally, the definition is as follows.
\begin{definition}[Rumor and Truth Frontiers]\label{def:rumor_truth_frontier}
Consider the rooted information spread process on a tree $T = (V,E)$ with root $r$ and let $t > 0$. The \emph{rumor frontier} after round $t$ is the collection of all nodes $v \in V$ with $f_t(v) \neq \none$ and that either (i) have a (child) node $u \in V$ with $p(u) = v$, which satisfies $f_t(u) = \none$, or (ii) do not have children. 
Formally, the {\em Rumor Frontier} at time $t$, denoted $R_t$, of the spread process $\{f_t(v)\}$, is given by 
$$ R_t = \{v \in V\ |\  f_t(v) \in \{-1,1\} \  \text{and}\  \exists u \in V \text{ s.t. } f_t(u) = \none \ \text{and}\ v = p(u)\}.$$
The \emph{truth frontier} after round $t$ is the collection of all nodes $v$ satisfying the following: the shortest path $r = v_0 \to v_1 \to \ldots \to v_i = v$ from $r$ to $v$ satisfies $f_t(v_0) = f_t(v_1) = \ldots = f_t(v_i) = 1$, whereas either (i) there exists a child node $u \in V$ with $p(u) = v$ and $f_t(u) \neq 1$, or (ii) $v$ has no children. 
Formally, the {\em Truth Frontier} at time $t$, denoted $F_t$, of the spread process $\{f_t(v)\}$, is the unique set which satisfies the condition:
$$F_t = \{r \} \cup \{v \in V \setminus \{r\} \ |\  f_t(v) = 1 \ \text{and}\  p(v) \in F_t \ \text{and}\  \exists u \in V \text{ s.t. } f_t(u) \ne 1  \text{ and } v = p(u)\}.$$
\end{definition}
Observe that both the rumor frontier and the truth frontier always progress away from the root. Clearly, the rumor frontier after round $i$ contains all nodes at distance exactly $i$ from the root, but at what rate does the truth frontier move? is the rumor frontier substantially faster than the truth frontier? how do opinions typically behave between these two frontiers?

\paragraph{The rooted path}
A special case of particular importance is that of a rooted path, which may either refer to the \emph{finite} path
$$
\cP_n = u_0 \to u_1 \to \ldots \to u_n
$$
where the root is $r=u_0$, or the \emph{infinite path} with the same root:
$$
\cP_\infty = u_0 \to u_1 \to \ldots \to u_n \to \ldots
$$
In general, the results for paths in this paper apply both to the finite or the infinite path, where in the finite case we typically think of $n$ as a large parameter, whereas the model parameters $\alpha$ and $\beta$ are thought of as fixed (independent of $n$).
The following simple lemma provides a lower bound on the speed of the truth frontier for the rooted path $\cP_n$, where $n$ can be finite or $\infty$. Note that in this case, both the truth frontier and the rumor frontier at any given time are of the form $\{v_i\}$ for some $i$ (possibly different between the two frontiers), and with slight abuse of notation we sometimes write that the node $v_i$ is the frontier. %\omri{Perhaps results that hold for paths/trees should take a tuple of only four parameters, not five. This can reduce confusion a bit.}
%\madhu{I think we can stick with the current format - first occurence of the path is the graph, and second is the tree .... and parameters always mention both, though if the graph is a tree then obviously the second occurence is redundant.}

\begin{lemma}[Lower bound for truth frontier]\label{lem:lb_truth_frontier}
%\omri{Need to make this more concise.}
Fix $a,b \in (0,1)$. Consider the spread process $\{f_t(v)\}$ associated with parameters $(\cP_n, r, \cP_n, a, b)$, where $n$ may be finite or infinite.
%The truth frontier after round $t$ is $\{v_i\}$, where the distribution of $i$ dominates $\min\{\Bin(t-1,b), n\}$. In particular,
%$i \geq \min\{(b - o(1))t, n\}$ with probability $1 - e^{-t^{\Theta(1)}}$. The $o(1)$ term tends to zero at a polynomial rate as $t \to \infty$.
%Formally, letting $\cP = v_0 \to v_1 \cdots \to v_n$ denote the path of length $n$ with root $r=v_0$ and $p(v_i) = v_{i-1}$ specify the unique tree in this graph, we have that for every $a,b \in (0,1)$ the spread process $\{f_t(v)\}$ associated with $(\cP,r,T,a,b)$ satisfies the following: 
For every $t$, the truth frontier $F_t$ is a singleton, and $F_t = \{v_i\}$ where the distribution of $i$ dominates $\min\{\Bin(t-1,b),n\}$.\footnote{For two real valued random variables $X$ and $Y$, we say that $X$ dominates $Y$ if for every $v$ it holds that $\Pr[X \geq v] \geq \Pr[Y \geq v]$.} In particular,
$i \geq \min\{(b - o(1))t, n\}$ with probability $1 - e^{-t^{\Theta(1)}}$. The $o(1)$ term tends to zero at a polynomial rate as $t \to \infty$.
\end{lemma}

For completeness we provide the proof of this easy lemma below. 
\begin{proof}
Consider a different process on the path with values $1$ and $\none$. 
At time $0$, $v_0$ has the value $1$ and all other nodes the value 
$\none$. At step $t$ of the process, let $i(t)$ be the maximal value where 
$v_i = 1$. Then let $v_{i+1} = 1$ with probability $b$ and $v_{i+1} = \none$ with probability $1-b$. All other nodes do not change. 
It is clear that $i(t) \sim Bin(t,b)$. It is also clear that the truth frontier in the process $f_t$ dominates $i(t-1)$. The second statement of the lemma follows from Chernoff bound. 
\end{proof}

The above arguments cannot easily yield an upper bound on the truth frontier, as there might exist additional mechanisms which accelerate the progress of this frontier in our process. Indeed, it may be the case that the current truth frontier meets and ``swallows'' consecutive runs of nodes holding a $+1$ value, thereby substantially pushing the truth frontier forward. 
In fact, it is not clear a priori whether the rumor and truth frontier are separated at all. In the next subsection, we conduct a thorough analysis of the process. We show, among other results, that the lower bound on the truth frontier speed is in fact tight (up to lower order terms); along the way, we obtain interesting characterizations of the behavior of nodes that lie between the truth frontier and the rumor frontier at a given time.

\section{The Decoupled Framework and its Analysis}
\label{sec:analysis_path}

%\madhu{Lots of mention of opinion dynamics in this section (and maybe elsewhere). Should make a global pass to change to "spread process".}

At the core of our analysis is an object called a \emph{run}. In the directed path, this is a collection of consecutive nodes along the path that all have the same opinion at a given time. The more general definition in trees involves a connected component of nodes having the same value.

\begin{definition}[Connectivity-based run]\label{def:connectivity_run}
%Consider the information spread process on a graph $G$, let $t \geq 0$, and let $T$ be the tree consisting of all opinionated (non-$\none$) nodes at time $t$. For any node $v$ with $f_t(v) \neq \none$, the \emph{run} containing $v$ is the maximal connected subgraph $C \subseteq T$ with $v \in C$, satisfying $f_t(u) = f_t(c)$ for any $u \in C$. 
Consider the spread process $\{f_t(v)\}_{t,v}$ with parameters $(G,r,T,a,b)$. The \emph{run} containing $v$ at time $t$, denoted $\run_t(v)$, is defined as the maximal connected component of $T$ satisfying $f_t(v) = f_t(u)$ for every $u \in \run_t(v)$. %\madhu{why not use the same definition for all graphs $G$, and not just tree?} 
The {\em value} of $\run_t(v)$ is defined as $f_t(v)$. 
\end{definition}

Clearly, the truth frontier does not equal the rumor frontier at a given time in the process if and only if a $(-1)$-valued run exists at this point. 

How do runs form and evolve throughout the process? How likely are they to survive and what is their typical structure? 
%For simplicity, let us focus again on the directed path $r = v_0 \to \ldots \to v_n$. A new run is created at the rumor frontier during ``error'' updates, when a node which currently has no opinion forms the opposite of its parent opinion (with probability $a$). It may then extend as long as newly formed opinions agree with their corresponding parent at the rumor frontier (with probability $1-a$), and can possibly merge with (older) runs 
addressing such questions is essential to understand our spread process. Unfortunately, the behavior and structure of each particular run in our process is heavily dependent on the behavior of other runs, and the dependencies seem difficult to analyze. Among other challenges, it is not immediately clear how to analyze the fact that runs can ``swallow'' one another along the process.

In order to mitigate these issues and simplify the analysis, we establish results showing that (i) our spread process is equivalent to a ``decoupled'' model with its own version of the truth frontier, and a modified notion of runs which can evolve with time, but do not merge; and (ii) while the behavior of any particular run is heavily dependent upon other events throughout the process, this can be mitigated by considering intermediate frontiers, that lie between the truth and the rumor frontier. As we show, frontiers whose starting points are separated away are unlikely to merge quickly, which implies that at least one (decoupled) run survives between them.

\subsection{The decoupled model}
\label{sec:decoupled_model}
One feature of runs in our spread process over a tree (Definition \ref{def:rooted_tree}) is that they can merge with time, a feature that is seemingly difficult to analyze. Consider for instance the case where $T$ is the rooted path $r = v_0 \to \ldots \to v_n$ and observe some point in the process where there are three consecutive runs of lengths 100, 1, 100 (note that their corresponding values are either $+1, -1, +1$ or the symmetric $-1, +1, -1$). If the single node $u$ in the central run changes its opinion in the current round to the parent's opinion (with probability $b$), while its child does not change its opinion (with probability $1-b$), then the two long runs merge into a longer run of length at least $200$. 
To overcome this, we propose and analyze another tree model in which runs do not merge and the value associated with a run is independent from those of other runs (for example, two consecutive runs in the path may have the same value, but are still considered separate runs). In the modified model, instead of updating a node's own opinion to its parent's with some probability, the node may decide (with a different probability) to pick its updated value as a result of an independent coin flip.

\begin{definition}[Information spread process over tree with explicit decoupled runs]
\label{def:rooted_tree_decoupled}
Fix parameters $\alpha, \beta \in [0,1]$ and a tree $T$. As in Definition \ref{def:rooted_tree}, each (non-root) node $v$ has a fixed parent $p(v)$. Additionally, every node holds an opinion $g_t(v) \in \{+1, -1, \none\}$ at time $t$. As before, the root always has opinion $+1$, and at time $t=0$, all other nodes have opinion $\none$ initially. 
In addition, we maintain the following information:
\begin{itemize}
    \item A collection $\bar Z = Z_0, Z_1, Z_2, \ldots$ where $Z_0 = 1$ and the other $Z_i$ values are picked independently and uniformly at random from $\{-1, 1\}$. Intuitively, $Z_i$ is the value of run number $i$, where smaller $i$ corresponds to an older run, and the first run always represents the correct value. For convenience, we set $Z_\infty = \none$.
    \item A counter, $\COUNT_t$, initialized to $\COUNT_0 = 1$; the value of $\COUNT_t$ at any given time $t$ is the amount of values from the sequence $\bar Z$ observed so far, that is, the number of runs we have ever had. %\madhu{I think it will be good to index $\COUNT$ by time also and expand the definition here to allow for this change of notation.}
    \item For each node $v$ at time $t$, we define a function $\origin_t(v)$, which maps $v$ to an index $j$ such that $g_t(v) = Z_j$. Initialize $\origin_0(r) = 0$ and $\origin_0(v) = \infty$ for all $v \neq r$. This is a mapping from a node to the run it currently belongs to.
\end{itemize}
The update rule for any $t \geq 0$ is defined as follows: 
\begin{enumerate}
    \item For any node $v$ with $g_t(v) = g_t(p(v)) = \none$, we set $g_{t+1}(v) = \none$ and $\origin_{t+1}(v) = \infty$.
    \item Fix an arbitrary ordering over all nodes $v$ with $g_t(v) = \none$ and $g_t(p(v)) \neq \none$. For any such $v$ (sequentially according to the ordering): 
    \begin{itemize}
        \item For $j = \COUNT_t$, we set
            \[
            g_{t+1}(v) = \begin{cases}
            g_{t}(p(v))  &\text{w.p. } 1-\alpha,
            \\
            Z_{j} &\text{w.p. } \alpha.
            \end{cases}
            \]
        \item In the first case, we set $\origin_{t+1}(v) = \origin_t(p(v))$, meaning that by taking the last value of $p(v)$, node $v$ joins the same run that $p(v)$ has most recently been in. We also set $\COUNT_{t+1} = \COUNT_t$.
        \item The second case corresponds to the creation of a new run; we subsequently set $\origin_{t+1}(v) = \COUNT_{t}$ and $\COUNT_{t+1} = \COUNT_t+1$.
    \end{itemize}
    \item For any node $v$ already having an opinion, i.e., $g_t(v) \neq \none$,
    \[
g_{t+1}(v) = \begin{cases}
g_{t}(p(v))  &\text{w.p. } \beta,
\\
g_{t}(v) &\text{w.p. } 1-\beta,
\end{cases}
\]
where in the first case we set $\origin_{t+1}(v) = \origin_t(p(v))$ (that is, $v$ joins the parent's run) and in the second case we set $\origin_{t+1}(v) = \origin_t(v)$ (meaning that $v$ remains in the same run).
\end{enumerate}
Formally thus a {\em decoupled process} has parameters $(G,r,T,\alpha,\beta)$ and is given by $(\barg=\{g_t(v)\}_{t,v}, \barZ= \{Z_\ell\}_\ell,  \barCOUNT = \{\COUNT_t\}_t, \barorigin = \{\origin_t(v)\}_{t,v})$ with $t,\ell \in \N$, $v \in V$ and with $g_t(v) \in \{-1,\bot,1\}$, $Z_\ell \in \{-1,1\}$, $\COUNT_t \in \N$ and $\origin_t(v)\in \N$ for every $\ell,t,v$. 
\end{definition}

\begin{definition}[Decoupled run]
%Consider the decoupled tree spread process from Definition \ref{def:rooted_tree_decoupled}. The $\ell$-th run at time $t$, denoted $R_\ell(t)$, is defined as the collection of all nodes $v$ with $\origin_t(v)=\ell$.
For a decoupled process $\dcpl$ as given in 
 %\cref{def:rooted_tree_decoupled} 
 Definition \ref{def:rooted_tree_decoupled}
x and $\ell,t \in \N$, the $\ell$-th \emph{run} at time $t$, denoted $R_\ell(t)$, is defined as the collection of all nodes $v$ with $\origin_t(v)=\ell$. The \emph{creation time} of $R_{\ell}$ is the minimum $t$ for which $R_{\ell}(t)$ is non-empty.
\end{definition}

In the next observation we state a few basic properties of decoupled runs.
To this end, let us recall some standard notions on trees. The branch leading to a node $u$ is the unique simple path in the tree from the root $r$ to $u$. Any node $v \neq u$ on this path is said to be an ancestor of $u$, and correspondingly, $u$ is a descendent of $v$.   

%\madhu{Define also a branch. I am not sure what it is in item 4 below.}

\begin{lemma}[Properties of decoupled runs]
%For every $t \geq 0$ and $v \in V$ the following holds: 
%\begin{itemize}
%    \item $g_{t}(v) = Z_\ell$ where $\ell = {\origin_t(v)}$. 
%    \item If $u$ is an ancestor of $v$ in $T$, then $\origin_t(u) \leq \origin_t(v)$. \madhu{In general avoid "any" in a math paper. It has an ambiguous meaning that sometimes means "for every" and sometimes means "there exists". E.g., ``If any one in the room does not speak Hebrew, we'll switch to English." :-)}
%\item For any $t' > t$, we have $\origin_{t'}(v) \leq \origin_t(v)$.
%    \item The induced subgraph of the tree $T$ on $R_{\ell}(t)$ is not necessarily connected. However, the intersection of $R_{\ell}(t)$ with any branch of $T$ is connected.
%\end{itemize} 
For every sequence of parameters $(G,r,T,\alpha,\beta)$ and decoupled process $\dcpl$ with these parameters, and every $t\in\N$ and $v \in V$, the following hold:
\begin{itemize}
    \item $g_{t}(v) = Z_\ell$ where $\ell = {\origin_t(v)}$. 
    \item For every ancestor $u$ of $v$, it holds that $\origin_t(u) \leq \origin_t(v)$.
    \item For every $t' > t$, we have $\origin_{t'}(v) \leq \origin_t(v)$.
    \item The induced subgraph of the tree $T$ on $R_{\ell}(t)$ is not necessarily connected. However, the intersection of $R_{\ell}(t)$ with any branch of $T$ is connected. 
\end{itemize} 
\end{lemma}
The proof of most statements follows in a straightforward manner from the model definition and is thus omitted. To see the first statement of the last part, observe for example that it may be possible for two nodes $u, v$ with $p(u) = p(v) = w$ to both have the same origin at time $t$ as that of their parent $w$, whereas $w$ may at some time $t' \geq t$ take a different origin; it follows that at time $t'$, these nodes $u,v$ will be in a run not containing $w$.

%Recall that the $\ell$-th run at time $t$, $R_\ell(t)$, constitutes of all nodes $v_i$ with $\ell = \origin_t(i)$. By the above lemma, either $R_\ell(t)$ is an empty set, or these nodes are consecutive in the path, that is, there exist $i_1, i_2$ so that
%$R_\ell(t) = \{v_{i_1}, v_{i_1+1}, \ldots, v_{i_2}\}$.

Note that in the new model, runs are created only at the rumor frontier, whereas each node $v$ moves with time to runs that are older and older (until eventually joining the first run, which happens once the truth frontier passes by $v$). Unlike connectivity-based runs, neighboring decoupled runs may hold the same value.
While the new model is quite a bit more complicated to define than the original one, we show that it is (a) equivalent to the original model, in the sense that the corresponding probabilistic processes can be coupled; and that (b) its description is sufficiently expressive for us to be able to carefully track the behaviour of runs throughout the process.
The next lemma formally states the equivalence as a coupling between the processes.

\begin{lemma}[Equivalence between models]
\label{lem:equivalence_models}
%There exists a coupling between the models described in Definitions \ref{def:rooted_tree} and \ref{def:rooted_tree_decoupled}, which matches each instantiation of the first model parametrized by $a,b$ and with values $f_t(v)$ to an (equal probability) instance of the second model parametrized by $\alpha = 2a, \beta = b$, with values $g_t(v)$, where $f_t(v) = g_t(v)$ for every $t$ and $v$.
For every parameter setting $(G,r,T,a,b)$ with $0 \leq a \leq 1/2$ there is a coupling between the spread process $\{f_t(v)\}_{t,v}$ with parameters $(G,r,T,a,b)$ and the decoupled process
$\dcpl$ with parameters $(G,r,T,\alpha=2a,\beta=b)$ satisfying $g_t(v) = f_t(v)$ for all $t$ and $v$.
\end{lemma}
\begin{proof}
In both models, the value of the root always equals one; furthermore, in both models, if $v$ and $p(v)$ have the same value at time $t$, then $v$ will still hold that value in time $t+1$. 

By induction, it remains to consider two situations: one where $p(v)$ has an opinion and $v$ has no opinion at time $t$; and another where they hold contrasting (non-$\none$) opinions at time $t$.
In the first situation, the probability in the original model for $f_{t+1}(v) = f_t(p(v))$ is $1-a$. In the second model, we have $g_{t+1}(v) = g_t(p(v))$ either if (i) $v$ picks its parent's value with probability $1-\alpha = 1-2a$, or (ii) $v$ picks a new value, $Z_{\COUNT}$, with probability $2a$, and the said value equals $g_t(p(v))$ with probability $1/2$ (independently of the rest of the process). Therefore, the probability for  $g_{t+1}(v) = g_t(p(v))$ is $1-2a + 2a \cdot 1/2 = 1-a$.

In the situation where $f_{t}(v) = -f_t(p(v)) \neq \none$, the first model sets $f_{t+1}(v) = f_t(p(v))$ with probability $b$, while the second model does so with probability $\beta = b$.
\end{proof}

Note that the truth frontier of the decoupled process (which consists at time $t$ of every node $v$ that satisfies $\origin_t(v) = 0$ and $\origin_t(u) > 0$ for some neighbor $u$ of $v$) is dominated by the truth frontier of the original process. The proof is similar to that of Lemma \ref{lem:lb_truth_frontier}.

\subsection{Runs and Their Characteristics}
We now turn to address the aforementioned fundamental questions about (decoupled) runs: what is their probability to survive? what is their typical structure along the process?   

The following lemma clarifies the importance of the rooted path example, not just as an illustrative special case, but as a way to analyze the behavior along any specific branch in a (possibly complicated) tree $T$. The lemma asserts that, when considering a general tree $T$, the behavior of runs along any particular branch $B$ of the tree is in a sense isomorphic to their behavior over the rooted path $\cP_{|B|}$ (where $|B|$ is the number of edges in the branch).
\begin{lemma}
\label{lem:run_branch_path}
Let $r \to v_{i_1} \to v_{i_2} \to \ldots \to v_{i_m}$ be a branch of a rooted tree $T$ on which we run the decoupled model (Definition \ref{def:rooted_tree_decoupled}) with parameters $\alpha$ and $\beta$. Consider another instantiation of the same model (with parameters $\alpha, \beta$) when run on the rooted path $\cP_{m}$ with vertices $u_0 \to u_1 \to \ldots \to u_m$, and let $t > 0$. For every pair $0 \leq j < j' \leq m$, the probability that $v_{i_j}$ and $v_{i_{j'}}$ are in the same (decoupled) run after round $t$ in the branch setting is equal to the probability that $u_j$ and $u_{j'}$ are in the same (decoupled) run after round $t$ in the path setting. The results also hold when $m=\infty$, i.e., when the branch and path are both infinite.
\end{lemma}
The lemma is quite trivial. We include a formal proof for completeness sake. 
\begin{proof}
We prove by induction that there is a coupling between the two settings, where each state of the branch (first setting)
is coupled to a state of the path (second setting) so that the following holds: $v_{i_j}$ and $v_{i_{j'}}$ belong to same run in $T$ after round time $t$ if and only if $u_j$ and $u_{j'}$ belong to the same run in $P_m$ after round $t$. 
The base case is trivial, as at time $t=0$, in both settings the root is in run $0$ and all other nodes are in run $\infty$ (recall that this means that their $\origin$ is $\infty$, and their value is $\none$).

Now suppose the validity of the statement after $t$ rounds and fix such a coupled pair of a branch state and a path state. In particular, both the branch and the path contain at this point the same number $\ell = \ell_t$ of runs with the same lengths. Denote these runs in the branch, ordered from the root outward, by $R_B^1, \ldots, R_B^\ell$, and the corresponding runs in the path by $R_P^1, \ldots, R_P^\ell$. For every $0 \leq j \leq m$, observe that $v_{i_j}$ has an opinion if and only if $u_j$ has one. Consider all possible cases:
\begin{itemize}
    \item If $v_{i_j}$ is in the same run as $v_{i_{j-1}}$ (and so also $u_j$ and $u_{j-1}$ satisfy this), then $v_{i_j}$ and $u_j$ will remain in their current run at time $t+1$.
    \item If $v_{i_j}$ is the first node in its run $R_B^k$, where $k < \infty$ (that is, $v_{i_j}$ already has an opinion) then its probability to move to the previous run is $R_B^{k-1}$ is $\beta$, and otherwise $v_{i_j}$ remains in the same run $R_{B_k}$. The same is true for $u_{i_j}$ -- its current run is $R_P^k$, and it moves backward to $R_P^{k-1}$ with probability $\beta$, and otherwise stays.
    Finally, if $v_{i_j}$ is the first node in the $\infty$ run (which means that so is $u_j$), then $v_{i_j}$ joins run $R_B^{k-1}$ with probability $1-\alpha$, and otherwise forms a new run. The same is true for $u_j$.
\end{itemize}
Therefore, the distribution over run structure after round $t+1$ is identical for both settings, provided that at time $t$ the runs had the same structure. This completes the inductive proof.
\end{proof}

We next wish to analyze the typical structure of runs over a rooted path $P_n$. With the help of Lemma \ref{lem:run_branch_path}, this will later on allow us to argue that opinions of vertices which are far away from each other in the tree are nearly uncorrelated.

\paragraph{Decoupled runs as random walks}
Consider a typical run $R_{\ell}$ in our decoupled model over the path $\cP_n = v_0 \to v_1 \to \ldots \to v_n$ (or $\cP_\infty$) with root $v_0$. How likely is it to survive? how large can it become?
Note first that a new run $R_\ell$ is created at the rumor frontier with probability $\alpha$ in each round $t \leq n$ (independently of the model behavior so far). In the first phase after having been created, the far end of the run is the rumor frontier of the process; thus, it adds an element each round with probability $1-\alpha$. On the near end of the rumor frontier, the first element of the run has probability $\beta$ in each round to take its parent opinion and leave the current run. 

Now, consider the first phase for $R_\ell$ as completed in the first time (after this particular run's formation) that a no-opinion node at the rumor frontier forms a new run; this is a probability $\alpha$ event in each round. It is not hard to see that the probability for the run to survive this phase is $1 - \alpha \beta$. Condition on this event. In the second phase of the run, in each round the first node of it can join the previous run (thus leaving $R_\ell$) with probability $\beta$, while the first node \emph{after} the run joins $R_\ell$ with probability $\beta$. This means, effectively, that the run length assuming it reaches the second phase is distributed according to a \emph{uniform random walk} over the line (with nonzero probability, specifically $1-2\beta(1-\beta)$ to remain at the same location). The initial location of the walk is some positive value (which with large constant probability is bounded by $C/\alpha$ for an absolute constant $C$), and we consider the run as \emph{killed} as soon as the walk has reached the origin, i.e., the value $0$, during the process.

The following lemma summarizes the above discussion on lengths of runs as random walks.
\begin{lemma}[Run length as random walk]
Fix $\alpha, \beta \in (0,1)$. Consider the decoupled process $\dcpl$ with parameters $(\cP_n,r,\cP_n,\alpha,\beta)$, where $n$ may be finite or $\infty$. For $\ell > 0$, suppose that the run $R_\ell$ has creation time $t_0$. Then all of the following hold.
\begin{itemize}
    \item Write $z_\ell = |R_\ell(t_1)|$, where $t_1$ is the creation time of the next run, $R_{\ell + 1}$. Then $\E[z_\ell] = O(1)$. Moreover, $z_\ell > 0$ with probability bounded away from zero. Finally, $z_\ell < n^{0.1}$ with probability $1-\exp\left( -n^{\Theta(1)}\right)$. 
    \item Conditioning on the event that $z_\ell > 0$, the distribution of $|R_\ell(t)|$ for $t_1 \leq t \leq n$ is equal to that of an unbiased random walk on $\N \cup \{0\}$ with the following properties: (i) the walk starts at $z_{\ell}$ at time $t_1$; (ii) in any given time, it either increases by one with probability $c'$, decreases by one with the same probability $c'$, or stays put; and (iii) if the walk reaches $0$ at some point, it terminates.
\end{itemize}
\end{lemma}
Unbiased or uniform random walks over the line are among the most well studied objects in the probability theory literature. It is well known, see e.g.~\cite{Ritter:81}, that the probability of a uniform random walk starting (say) at location $O(1)$ to reach $0$ in $m$ rounds is $1-\Theta(\sqrt{m})$, and that conditioning on this not happening, the largest value that the walk will reach is typically of order $\Theta(\sqrt{m})$. This is easy to see with logarithmic correction since with high probability the walk will never reach $C \sqrt{m} \log m$ and then one can use Gamblers Ruin. 

For fixed $\alpha, \beta$, our information spread process over $P_n$ runs for $O(n)$ rounds, and since a new run is created typically every constant number of rounds, the above reasoning suggests that in intermediate steps of the process (after $\Theta(n)$ steps have passed and with $\Theta(n)$ steps left to go), typically $\tilde \Theta(\sqrt{n})$ runs of length up to $\tilde O(\sqrt{n})$ should exist. 

However, this naive reasoning does not take into account the dependencies between the behaviors of different runs; an actual analysis of these dependencies seems non-trivial. For example, if $R_{\ell}$ is known to survive after a certain amount of steps, how does it effect the structure and behavior of its neighboring run $R_{\ell+1}$?
%\madhu{Should we have had a proposition/lemma in this section ... saying process X is a random walk?}

\subsection{Intermediate frontiers}
With this in mind, we circumvent the need to analyze dependencies between individual runs by considering large enough unions of consecutive runs. Unlike individual runs, these unions have the useful property that they are very likely to survive and roughly maintain the same length throughout the process. In what follows we provide tools to bound the rates at which the starting points of each of these runs (which we call \emph{intermediate frontiers}) progress with time.

The analysis of intermediate frontiers we conduct here serves a central role in the proofs of essentially all main results in this paper.
Specifically, we prove two lemmas characterizing the behavior of intermediate frontiers and the interaction between them. In Lemma \ref{lem:progress_of_k_frontiers} we show that these intermediate frontiers progress at a rate that is highly concentrated around its mean (specifically, the progress rate over $\tau$ rounds is typically $\beta \tau \pm O(\sqrt{\tau})$). In Lemma \ref{lem:subsequent_intermediate_frontiers}, we use this to show that frontiers whose creation times differ by $\tau$ are unlikely to meet within roughly the first $\approx \tau^2$ rounds after the creation of either of them.

%\begin{definition}[$k$-frontier]\label{def:k_frontier}
%Consider the decoupled model (Definition \ref{def:rooted_tree_decoupled}) over the rooted path 
%$v_0 \to v_1 \to \ldots \to v_n$, let $t > 0$, and let $R_0, R_1, \ldots, R_\ell, R_\infty$ be the collection of all runs ever created up to and including round $t$. For any $k \in [0,\ell]$, the \emph{$k$-frontier} $\sigma_k(t)$ of the process is the maximum $i \in [0,n]$ for which $v_i \in \bigcup_{j=0}^{k} R_j$. Finally, for $0 \leq t < t'$ denote $\Delta_k(t,t') = \sigma_k(t') - \sigma_k(t)$.
%\madhu{I am not sure I understand the definition - but below is a guess at what it means. In particular I don't understand what $R_\infty$ is. }
%\end{definition}

\begin{definition}
[$k$-frontier of a path]\label{def:k_frontier}
Let $\dcpl$ be the decoupled process with parameters
$(\cP_n,r,\cP_n,\alpha,\beta)$ where $n$ is either finite or $\infty$. For $t,k \in \N$ we define  the \emph{$k$-frontier} of the process, denoted $\sigma_k(t)$, to be:
$$\sigma_k(t) = \max \{ i \in \N \ |\ i \leq n,\  \exists j \in \{0,\ldots,k\} \text{ s.t. } v_i \in R_j(t)\}.$$
 Finally, for $0 \leq t < t'$ we define the {\em $k$-frontier growth in $[t,t']$} to be the quantity $\Delta_k(t,t') = \sigma_k(t') - \sigma_k(t)$.
\end{definition}

Equivalently, the $k$-frontier is the largest index of a node in $R_j$, for the largest $j \leq k$ for which $R_j$ is non-empty at time $t$.

Thus, the $0$-frontier in the decoupled tree model (Definition \ref{def:rooted_tree_decoupled}) is a lower bound for the truth frontier in the original tree model (Definition \ref{def:rooted_tree}). The rumor frontier at time $t$ corresponds to the $\ell$-frontier for the maximum $\ell$ for which $\sigma_\ell(t) \neq \sigma_{\ell-1}(t)$, and we say in this case that the \emph{index} of the rumor frontier is $\ell$. We consider the $k$-frontiers for $0 < k < \ell$ as \emph{intermediate frontiers} and use concentration bound on the progress of certain frontiers to prove the main results.

The next technical lemma asserts that the progress rate of frontiers is highly concentrated.
For simplicity, the statement is given for the infinite path $\cP_\infty$.

%\madhu{Nitpicking: We define frontiers only for finite paths? We are proving something only for the infinite path. Need to resolve this.} \omri{This is being addressed. (Keeping comment as a reminder for myself.)}
\begin{lemma}[Progress rate of $k$-frontiers]
\label{lem:progress_of_k_frontiers}
%Consider the decoupled model (Definition \ref{def:rooted_tree_decoupled}) with parameters $\alpha, \beta$ over $P_\infty$, let $0 \leq s < t$, and set $\mu = \beta(t-s)$. Then
%$\mathbb{E}[\Delta_k(s,t)] = \mu + O\left(\frac{1}{\alpha}\right)$. Moreover 
%for any $0 < \delta \leq \mu$ we have
%$$
%\Pr\left(
%\bigg| \Delta_k(s,t) - \mu \bigg| >  \delta \sqrt\mu + \frac{\delta^2}{3\alpha} + 1 \right) < 2\exp\left(-\frac{\delta^2}{3}\right).  
%$$
%where the $o(1)$ term tends to zero as $\mu \to \infty$. 
%Furthermore, when the index of the truth frontier at time $s$ is not $k$, $\mathbb{E}[\Delta_k(s,t)] = \mu$ and the above bound improves to
%$$
%\Pr\left(
%\bigg| \Delta_k(s,t) - \mu \bigg| >  \delta \sqrt\mu \right) < \exp\left(-\frac{\delta^2}{3}\right).
%$$

%Let $\cP_\infty$ denote the infinite path with edge set $\{v_i\to v_{i+1} | i \in \N\}$, and let $r=v_0$ be its root.
For every decoupled process $\dcpl$ with parameters $(\cP_\infty,r,\cP_\infty,\alpha,\beta)$ and for every $k,s,t \in \N$ with $0\leq s < t$ we have $\mathbb{E}[\Delta_k(s,t)] = \mu + O\left(\frac{1}{\alpha}\right)$ where $\mu := \beta(t-s)$. 
Moreover 
for every $0 < \delta \leq \mu$ we have
$$
\Pr\left(
\bigg| \Delta_k(s,t) - \mu \bigg| >  \delta \sqrt\mu + \frac{\delta^2}{3\alpha} + 1 \right) < 2\exp\left(-\frac{\delta^2}{3}\right).  
$$
%where the $o(1)$ term tends to zero as $\mu \to \infty$. 
Furthermore if the index of the rumor frontier at time $s$ is larger than $k$ %\madhu{I am not sure what index means - perhaps express this in terms of $R_\ell(t)$?} 
then $\mathbb{E}[\Delta_k(s,t)] = \mu$ and the above bound improves to
$$
\Pr\left(
\bigg| \Delta_k(s,t) - \mu \bigg| >  \delta \sqrt\mu \right) < \exp\left(-\frac{\delta^2}{3}\right).
$$
\end{lemma}
The above lemma implies, in particular, that in $\cP_\infty$, that 
a frontier that was in location $l$ at time $t_0$ will typically be in location $$l + \beta \tau \pm O(\sqrt \tau)$$ at time $t_0 + \tau$. In the finite case, i.e., when considering the finite path $\cP_n$, the behavior is similar but capped at location $n$: specifically, as long as 
$l + \beta \tau \leq n - \Theta(\sqrt{\tau})$, the same result holds for $\cP_n$; once this is not true anymore, the location of the frontier is typically between $n - \Theta(\sqrt{\tau})$ and $n$.

%In the finite case, i.e., when one replaces $\cP_\infty$ in the statement of the lemma with $\cP_n$, the precise statements may become slightly messier. For the applications in this paper, the easiest way to adapt the above results to the finite case is through the high probability concentration bounds. Specifically, one can ``imagine'' that we are located on the infinite path, show that the location of the $k$-frontier is well-concentrated around its mean, and if the mean is sufficiently smaller than $n$, then the concentration result carries to $\cP_n$. \omri{I feel this paragraph can/should be clarified.}

\begin{proof}[Proof of Lemma~\ref{lem:progress_of_k_frontiers}]
Consider first the case where index of the truth frontier at time $s$ is not $k$, that is, that the $k$-frontier is not the rumor frontier. In this case, the probability for the $k$-frontier to increase by one in any given round is $\beta$, and so $\Delta_k(s,t) \sim \Bin(t-s, \beta)$, with mean $\mu = \beta(t-s)$ 
(note that the $O(1/\alpha)$ term in the statement of the lemma is simply zero here). From a standard multiplicative Chernoff bound, we have
$$
\Pr(|\Delta_k(s,t) - \mu| > \delta \sqrt{\mu}) < \exp\left(-\frac{\delta^2}{3}\right),
$$
which completes the proof for the first case.

In the case that $R_k$ is the rumor frontier, denote by $s' > s$ the first round in which a subsequent non-$\infty$ run, $R_{k+1}$, is formed. Then $s' - s$ is distributed according to the geometric distribution $\Geom(\alpha)$, and so 
$$
\Pr(s'-s > \zeta) = (1-\alpha)^\zeta \leq \exp(-\alpha\zeta).
$$
Setting $\zeta = \delta^2/3\alpha$ we get a probability bound of $\exp(-\delta^2 / 3)$. Condition on the last event not holding, that is, on $s'-s \leq \zeta$. In this case, the $k$-frontier starts by increasing by one in each of the first $s'-s-1$ rounds; then, it stays at the same location for one round; and in the remaining $t-s'$ rounds, the increase is $\Delta_k(s', t) \sim \Bin(t-s', \beta)$. The initial $s'-s$ rounds contribute at most $\zeta + 1$ to the deviation, and the analysis of the subsequent $t-s'$ rounds is similar to the first case.
\end{proof}

\paragraph{Keeping frontiers disjoint.}
How many disjoint intermediate frontiers will we see at time $t$? Here, the  $k$-frontier and $k'$-frontier are disjoint at time $t$ if $\sigma_k(t) \neq \sigma_{k'}(t)$. Lemma \ref{lem:progress_of_k_frontiers} indicates that the $k$-frontier for each $k$ is with high probability in location within $\pm \tilde O (\sqrt{t})$ of its expected location as was measured when the run $R_k$ was first created.

On the other hand, a new run is created typically every roughly $\approx 1/\alpha$ rounds of the process, and in particular, such intermediate frontiers are densely created at different points throughout the process. 
Intuitively, these facts may suggest that there are roughly, perhaps up to lower order terms, $\approx \sqrt{t}$ disjoint intermediate frontiers at time $t$, of distance up to $\approx \sqrt{t}$ from each other.\footnote{For frontiers that were created more recently -- at  round $s = t - o(t)$, where $t$ is the current round number -- the distance between neighboring disjoint frontiers that we can hope to obtain is of order $\approx \sqrt{t-s}$.}

The next lemma suggests that this picture is indeed largely accurate. For ease of reading, the function $\omega(m)$ in the statement of the lemma can be thought of as, say, $\text{polylog}(m)$ or $m^{0.1}$. We did not try to optimize the dependence in $\alpha$ or $\beta$ in the proof.

%\madhu{Should formalize lemmas+theorems coming below and later - I ran out of steam.}
%\madhu{Semi-relatedly - it would be nice to summarize what this section achieves, say at the first para of this section.} \omri{Done.} \madhu{To do that it would be nice if the final lemma of the section had some nice punchline. Does it? Can we say it in English?} \omri{I don't know if the final lemma is itself the punchline. Both lemmas have their uses.}

\begin{lemma}[Condition for frontiers to remain disjoint]
\label{lem:subsequent_intermediate_frontiers}
Fix $\alpha, \beta \in (0,1)$ and a non-decreasing function $\omega \colon \N \to \N$ that satisfies $\omega(m) \leq O(m^{0.5-c})$ for an absolute constant $c > 0$. There exist $C = C_{\alpha, \beta} > 0$ and $C' = C'_\beta > 0$ which satisfy the following for the decoupled process $\dcpl$ with parameters $(\cP_\infty,r,\cP_\infty,\alpha,\beta)$. Let $0 \leq s < s' \leq t-C_{\alpha, \beta}$, and let $k, k'$ be the indices of the rumor frontier after rounds $s, s'$ respectively. If 
\begin{equation}
\label{eq:distance_between_frontiers_if}
s'-s > C'_\beta \cdot \omega(t-s) \cdot \sqrt{t-s}
\end{equation}
then with probability at least $1-4e^{-(\omega(t-s))^2}$, it holds that
$$
\sigma_k(t) = s + \beta(t-s) \pm 4 \omega(t-s) \sqrt{\beta(t-s)} \quad \text{and} \quad 
\sigma_{k'}(t) = s' + \beta(t-s') \pm 4 \omega(t-s) \sqrt{\beta(t-s)},
$$
and
$$\sigma_{k'}(t) - \sigma_k(t) > \omega(t-s) \cdot \sqrt{t-s} > 0.$$ 
\end{lemma}
Again, a similar result holds for the finite path $\cP_n$ (for large $n$ and fixed $\alpha, \beta$) as long as the expected location of either frontier is not very close to $n$.

\begin{proof}
Apply Lemma \ref{lem:progress_of_k_frontiers} twice, once for each of the pairs $(s,t)$ and $(s', t)$. Take $\delta = \sqrt{3} \cdot \omega(t-s)$. It follows that 
with probability at least $1 - 4\exp(-\delta^2/3) = 1-4 \exp\left(-(\omega(t-s))^2\right)$, the following two statements hold as long as $t-s'$ is large enough (as a function of $\alpha$ and $\beta$):
$$
\sigma_k(t) = s + \beta(t-s) \pm 2 \delta \sqrt{\beta(t-s)} \quad \text{and} \quad 
\sigma_{k'}(t) = s' + \beta(t-s') \pm 2 \delta \sqrt{\beta(t-s)}.
$$
Subtracting the former from the latter, we conclude that 
$$
\sigma_{k'}(t) - \sigma_k(t) \geq (1-\beta)(s' - s) - 4\delta\sqrt{\beta(t-s)} > 0
$$
where the last inequality holds if
$$
s' - s > \frac{4\sqrt{\beta}}{1-\beta} \cdot \delta \cdot \sqrt{t-s}.
$$
Setting $C' = \frac{4\sqrt{3 \beta}}{1-\beta} + \sqrt{3}$ completes the proof.
\end{proof}

Showing that large collections of  disjoint intermediate frontiers exist at all times throughout the process is a central argument in our proofs. Concretely, if there exist at least $k$ different frontiers (including the first one, i.e., the decoupled truth frontier) at some time $t$, then there are $k-1$ runs between the decoupled truth frontier and the rumor frontier. By bounding $k$ and analyzing the characteristics of these runs we can conclude, among other results, that the actual truth frontier is not far from the decoupled one. This result is proved next. 

\section{The Truth Frontier and When it Arrives}\label{sec:truth_frontier} 

\subsection{Upper Bound for Truth Frontier}
\label{sec:upper_bound_truth_frontier}

Consider the original process over a tree $T$ (Definition \ref{def:rooted_tree}) at time $t$ and recall our lower bound of $(1-o(1))bt$ for the truth frontier at time $t$, Lemma \ref{lem:lb_truth_frontier}. We now prove a matching upper bound, by leveraging the equivalence between the original model and the decoupled one (Lemma \ref{lem:equivalence_models}), picking a collection of sufficiently far intermediate frontiers that are nonetheless relatively close to the truth frontier, and showing that with high probability a $(-1)$-run survived at time $t$ between at least one couple of neighboring frontiers.

\begin{lemma}
\label{lem:progress_of_all_frontiers}
Fix $\alpha, \beta \in (0,1)$ and $0 < c < 1/2$. The following holds for the decoupled process $\dcpl$ with parameters $(\cP_\infty, r, \cP_\infty, \alpha, \beta)$. With probability at least $1-e^{-(1+o(1))t^{2c}}$, there are $\Omega_{\alpha, \beta}(t^{0.5-c})$ disjoint frontiers that at time $t$ are in locations 
$$\beta t \; + \;  C_\beta \cdot \omega(t) \sqrt t \cdot i \;\pm\; O(\omega(t) \sqrt t),$$ 
for a constant $C_\beta$ depending on $\beta$, and $i=0,1,\ldots,\Theta_{\alpha,\beta}(t^{0.5-c})$.
\end{lemma}
%As usual, the same type of behavior (capped at $n$) exists in the finite path $\cP_n$ for large $n$.
\begin{proof}
The proof follows from Lemma \ref{lem:subsequent_intermediate_frontiers}. by taking $\omega(m) = m^{c}$ and picking a sequence $s_0, s_1, \ldots, s_l$ 
as follows: 
$$s_i = i \cdot C'_\beta \cdot \omega(t) \sqrt{t} + 1,$$ where $C'_\beta$ is as in the aforementioned lemma, and $l$ is the maximum possible for which $s_l \leq t/2$. Note that $l = \Omega_{\alpha, \beta}(t^{0.5-c})$. Now applying Lemma \ref{lem:subsequent_intermediate_frontiers} with parameters $s = s_i, s' = s_{i+1}$ and $t$ for $i=0,1,\ldots,l$, and taking a union bound over all applications of the lemma, we get that the frontiers $k_{0}, \ldots, k_{l}$ which constitute the rumor frontier at times $s_0, \ldots, s_l$ are all disjoint at time $t$; Furthermore, the location of every frontier $k_i$ at time $t$ is
$$
\sigma_{k_i} = s_i \;+\; \beta(t-s_i) \;\pm\; O(\omega(t) \sqrt{t}) \;=\; \beta t \;+\; (1-\beta) s_i \;\pm\xspace\; O(\omega(t) \sqrt{t}).
$$
Note that the union bound is over a number of events polynomial in $t$, which goes into the $o(1)$ term in the probability expression.
\end{proof}

We are now ready to prove the matching upper bound on the truth frontier for the original tree spread process (Definition \ref{def:rooted_tree}). 
\begin{theorem}[Upper bound for truth frontier]
\label{thm:ub_truth_frontier}
Fix $a \in (0,1/2)$ and $b \in (0,1)$. Consider the original spread process with parameters $(\cP_n, r, \cP_n, a, b)$, with $n$ being finite or $\infty$. The truth frontier after round $t$ is the singleton $\{v_i\}$, where $i \leq \min\{(b+o(1))t, n\}$ with probability $1-e^{-t^{\Theta(1)}}$. The $o(1)$ term tends to zero at a polynomial rate as $t \to \infty$.
\end{theorem}

\begin{proof}
Consider the case that $n = \infty$. The case of finite $n$ requires an easy modification to be described later.

We apply Lemma \ref{lem:equivalence_models} and consider the decoupled model with $\alpha = 2a$ and $\beta = b$ at time $t$. Now apply Lemma \ref{lem:progress_of_all_frontiers} with (say) $\omega(m) = m^{0.1}$. 
With probability at least $1 - e^{-t^{\Theta(1)}}$, there are $\Omega(t^{0.4})$ disjoint frontiers at time $t$, out of which at least $t^{0.2}$ are in location at most $\beta t + O(t^{0.8})$.
Consider any of the runs between the $(i-1)$-frontier and the $i$-frontier for $i \leq t^{0.2}$. By definition of the decoupled model, the value associated with the run is $-1$ with probability $1/2$, independently of any other randomness in the process. Therefore, there exists at least one ($-1$)-run among these with probability at least $1-2^{-t^{0.2}}$. Conditioning on this event, the truth frontier of the coupled original process is bounded by $\beta t + O(t^{0.8})$, which completes the proof.

To handle the case where $n$ is finite, we note that if $n \leq t + O(t^{0.8})$ then the statement is trivial (as the truth frontier cannot be larger than $n$), and otherwise, the above proof works word-for-word by ``imagining'' that the process runs over $P_{\infty}$ and discarding any frontier that crosses $n$.
\end{proof}

The last theorem combined with Lemmas \ref{lem:lb_truth_frontier} and \ref{lem:run_branch_path} immediately yields the following corollary about the behavior of the process along any particular branch in general graphs.

\begin{corollary}[Truth frontier along tree branch]
\label{coro:truth_frontier}
Fix $a \in (0,1/2)$ and $b \in (0,1)$. Consider the original spread process with parameters $(G, r, T, a, b)$, and let $B$ be a branch of $T$ (rooted at $r$) whose length $n$ may be finite or $\infty$.
The intersection of the truth frontier with the branch $B$ after round $t$ is a singleton $\{v\}$, where the distance between $r$ and $v$ satisfies $d_{G}(r,v) = \min\{(b \pm o(1))t, n\}$ with probability $1-e^{-t^{\Theta(1)}}$. The $o(1)$ term tends to zero at a polynomial rate as $t \to \infty$.
\end{corollary}

\subsection{Stabilization of individual opinion: When am I correct?}
\label{sec:when_am_i_correct}

Consider the spread process over a tree $T$ with root $r$, and let $v \neq r$ be some node. We know that the opinion $f_t(v)$ converges to one as $t \to \infty$, and with very high probability as soon as $t \approx (1+o(1)) b \cdot d_T(r,v)$, where $d_T(r,v)$ is the distance between $r$ and $v$ in the tree. 

However, it may not be entirely realistic to assume that each individual agent knows the structure of the communication network, and in particular its distance from the root. Is there a way for such an agent to infer, as soon as possible, once it has stabilized to the correct opinion in a graph-independent way?

Our next set of results shows that the answer is positive. When a node $v$ that received its initial opinion $t_v$ rounds ago keeps the same opinion for some $\tilde{\Theta}\left(\sqrt{t_v}\right)$ rounds, with high probability this is because the opinion of $v$ has converged to the correct one.

\begin{theorem}
\label{thm:changes_of_opinion_individual}
Fix $a \in (0,1/2)$ and $b \in (0,1)$. Consider the original spread process with paramaters $(G, r, T, a, b)$ and let $v \neq r$ be a vertex in the graph. For any $\eps > 0$, there exists a constant $C = C(a,b,\eps)$, which does not depend on $G, T, r, v$, so that with probability at least $1-\eps$, the following holds. In the first $q$ rounds after $v$ forms an initial (non-$\none$) opinion, if the truth frontier has not reached $v$ during these $q$ rounds, then the longest consecutive streak of rounds in which $v$ does not change its opinion is at most 
$$C_{a,b,\eps} \cdot q^{1/2} \cdot (\log q)^{3/2}.$$
\end{theorem}
\begin{proof}
The main idea is as follows. We consider the decoupled model and analyze the typical behavior of decoupled runs that go through $v$ as the process evolves, as a function of the number of rounds that $v$ has a non-$\none$ opinion so far. By bounding the size of the maximum decoupled run that has gone through $v$ up until any time of the process, and showing that with high probability no more than $O(\log n)$ consecutive decoupled runs have the same opinion, we can bound with high probability the behavior of runs in the original process as they go through $v$.

Let $d$ denote the distance from $r$ to $v$ and recall that $v$ receives its initial opinion at round $d$. Define $\omega = \omega_{\beta, \eps} \colon N \to N$ by 
$\omega(m) = 2 \sqrt{\log \frac{1} {\eps} + \log \frac{C'_\beta}{\beta} + \log{m}}$. %, for a large enough constant $C_{a,b,\eps} > 0$ depending only on $a,b,\eps$.
We pick round numbers $0 \leq s_l < s_{l-1} < \ldots < s_2 < s_1 < d$ satisfying the following.
\begin{itemize}
    \item $s_0 = d-C_{\alpha, \beta}^{\eqref{lem:subsequent_intermediate_frontiers}}$, where $C_{\alpha, \beta}^{\eqref{lem:subsequent_intermediate_frontiers}}$ is the corresponding constant depending on $\alpha$ and $\beta$ defined in the statement of Lemma \ref{lem:subsequent_intermediate_frontiers}.
    \item $s_i$ for $i > 0$ is chosen as the maximum value for which $s_{i-1} - s_i > C'_\beta \omega(d-s_i) \sqrt{d-s_i}$, where $C'_\beta$ is as defined in Lemma \ref{lem:subsequent_intermediate_frontiers}. Take $l$ to be the first value for which $s_{l+1} < 0$; reset $s_l$ to $0$, delete $s_{l+1}$ and stop the process.
\end{itemize}
For any $1 \leq i \leq l$, let $k_i$ denote the index of the rumor frontier at time $s_i$.

Now for any $i$ pick $t_i$ to be the minimum time $t$ satisfying that $s_i + \beta(t-s_i) - 4\omega(t-s_i) \sqrt{\beta(t-s_i)} \geq d$.
Note that $t_{i} - t_{i-1} = O_{\alpha, \beta, \eps}(s_{i-1} - s_i)$. Similarly pick $t'_i$ as the maximum time $t$ satisfying $s_i + \beta(t-s_i) + 4\omega(t-s_i) \sqrt{\beta(t-s_i)} \leq d$.
Note that both $t_{i} - t'_i$ and $t_i - t_{i-1}$ are of order $\Theta_{\alpha, \beta, \eps}(s_{i-1} - s_i) = \Theta_{\alpha, \beta, \eps}(\sqrt{(d-s_i) \cdot \log(d-s_i)})$.

Apply Lemma \ref{lem:subsequent_intermediate_frontiers} for all tuples $(s_i, s_{i-1}, t_i)$ and $(s_i, s_{i-1}, t'_i)$ and take a union bound over all applications of the lemma, concluding that the statement of the lemma holds for all tuples simultaneously with probability $1-\eps/2$. Conditioning on the last event, the total number of rounds between the time the $k_i$-frontier reaches $v$ and the time $k_{i-1}$-reaches it is positive and bounded by $O_{\alpha, \beta, \eps} \left( \sqrt{(d-s_i) \cdot \log(d-s_i)} \right)$.

So far, we have identified a collection of frontiers that all remain separated when they pass through $v$, and the number of rounds that $v$ spends between each pair of consecutive (decoupled) frontiers is not too large. 
This means that each decoupled run passing through $v$ is not very large by itself. But it is still theoretically possible that streaks of decoupled runs all having the same value ($+1$ or $-1$) will together form a large run in the original process. However, the remainder of the proof shows that such long streaks are unlikely to exist.

Specifically, for any $q$, consider the first $q$ rounds after $v$ forms an initial opinion. In this time frame, at most $q$ frontiers go through $v$. For each pair of such consecutive frontiers $k_{i-1}$ and $k_i$, pick an arbitrary decoupled run $R^{(i)}$ between them. Note that the value associated with this run is uniformly picked from $\{-1,1\}$. Therefore, with probability at least $1-\eps / 4 q^2$, there is no collection of more than 
$$O\left(\log \frac{1}{\eps} + \log q\right)$$ 
consecutive runs all holding the same value. Conditioning on this, and since each individual decoupled run goes through $v$ for at most $O_{\alpha, \beta, \eps} \left( q^{1/2} (\log q)^{1/2} \right)$ rounds, no run in the original process (which may concatenate decoupled runs as long as they have the same opinion) goes through $v$ for more than $O_{\alpha, \beta, \eps} \left( q^{1/2} (\log q)^{3/2} \right)$ rounds during the first $q$ rounds. The probability of the union of all of these events is at least
$$
1 - \frac{\eps}{4} \cdot \sum_{q=1}^{\infty} \frac{1}{q^2} \geq 1 - \frac{\eps}{2},
$$ 
which combined with the first part, completes the proof of the theorem.
\end{proof}

\begin{corollary}
\label{coro:when_to_stop}
Consider a one-player game in the same setting as in Theorem \ref{thm:changes_of_opinion_individual}, where an agent is located at vertex $v$, knows $a,b$ in advance, and can only observe the opinions of $v$ throughout the process;
The agent does not have any information about the graph $G$, tree $T$, or root $r$. The agent has the ability to end the game after any round of the process. Its goal is to end the game as soon as possible \emph{after} the truth frontier has reached $v$. For any $\eps > 0$, there exists a strategy for the agent with success probability $1-\eps$ to end the game no more than $C_{a,b,\eps} \sqrt{t} (\log t)^{3/2}$ rounds after the truth frontier has reached $v$, where $t$ is the total number of rounds that $v$ has been holding a non-$\none$ opinion so far.
\end{corollary}
We note that one cannot directly apply the statement of Theorem \ref{thm:changes_of_opinion_individual} to obtain the corollary, as each application of the theorem would incur a failure probability of $\eps$, and so a union bound based argument would not immediately work here. 
However, the corollary immediately follows from the \emph{proof} of the theorem. 
\begin{proof}[Proof of Corollary \ref{coro:when_to_stop}]
In the proof of the theorem we implicitly show the following two statements. (For simplicity, we omit dependencies in $a,b,\eps$ when not essential for the discussion.)
\begin{itemize}
    \item With probability at least $1-\eps/2$ the following holds for all $q$ simultaneously. Let $q$ denote the number of rounds since $v$ formed an initial opinion. Then the decoupled run that $v$ currently belongs is either the truth frontier, or it will remain at $v$ for at most $O(\sqrt{q \cdot \log q})$ rounds.
    \item Again, with probability at least $1-\eps/2$ the following holds for all $q$ simultaneously, where again $q$ is the amount of time since $v$ formed an opinion. The maximum number of consecutive decoupled runs with the same value passing through $v$ during these $q$ rounds is $O(\log q)$.
\end{itemize}
Let $t_{\text{truth}}$ denote the number of rounds from the time in which $v$ first forms an opinion until the truth frontier reaches $v$. These statements together imply that for any $q=1,2,\ldots,t_{\text{truth}}$, the agent will not see the same value for more than $O(q^{1/2} (\log q)^{3/2})$ consecutive rounds during the first $q$ rounds of the process since $v$ formed an opinion. Thus, the agent's strategy is very simple: suppose that $q \in \N$ rounds have passed since $v$ first formed an opinion, and maintain a counter of the number of rounds since $v$ last changed its opinion. If the value of the counter is more than some $\Theta(q^{1/2} (\log q)^{3/2})$, then end the game. 
\end{proof}
\iffalse
\begin{proof}
We use the exact same notation as in the proof of Theorem \ref{thm:changes_of_opinion_individual}, except that we replace $\omega$ defined in that proof with (say)
$$
\omega(m) = 2 \sqrt{\log \frac{1} {\eps} + \log \frac{C'_\beta}{\beta} + 3\log{m}}.
$$
First apply the second part of the proof of Theorem \ref{thm:changes_of_opinion_individual}, to conclude that with probability $1-\eps/2$, \emph{for all} $q \in \N$ the maximum number of decoupled runs in the first $q$ rounds after $v$ first forms an opinion is $O(\log 1/\eps + \log q)$. Condition on this. 

Apply Lemma \ref{lem:subsequent_intermediate_frontiers} with the same choices and definitions of $s_i, t_i$ for any $i$. It follows that

this implies that $(s_i, s_{i-1}, t_i)$ and $(s_i, s_{i-1}, t'_i)$ do not satisfy the conditions of the lemma with probability bounded by $O(\eps / m^2)$, where $m = d - s_i$. Note that
$s_0$ is known to the agent, and given $s_{i-1}$, the agent can compute $s_i$ without any knowledge on the graph structure.

We may fix all randomness of the process in advance (without revealing any of it to the agent) to \omri{unfinished.}
\end{proof}
\madhu{Need to write full proof}
\fi

\section{Unreliability of Newly Formed Opinions}
\label{sec:unreliability}
In this section we show that opinions in our process are unreliable, despite the fast convergence and the use of an error correction mechanism. 
Specifically, for graphs in which the neighborhoods around a vertex grow sufficiently quickly in terms of volume, the fraction of nodes holding the correct opinion among all nodes with a non-$\none$ opinion is $1/2 \pm o(1)$. 

\subsection{Proofs in expectation}
\label{sec:unrel_in_expect}
We start by proving the last statement in expectation over the randomness of the process. We state and prove this shortly, but the main idea is as follows. Consider any particular branch in the tree. With high probability, the truth frontier of the decoupled process advances at a rate of $b \pm o(1)$ along this branch with high probability (Lemmas \ref{lem:lb_truth_frontier}, \ref{lem:run_branch_path} and Theorem \ref{thm:ub_truth_frontier}). Conditioned on this event, the opinion of any non-$\none$ node that is further down the branch is uniformly distributed in $\{ 1, -1\}$. Given the growth in volume of the graph, most nodes are close to the far end of their branch, which proves the result.

We next formalize this reasoning, starting with a couple of definitions. The first definition, of the volume growth of a graph (or a tree), refers to growth patterns in the sizes of balls around a node $v$ in the graph (or tree).
\begin{definition}[Volume growth]
Let $G = (V, E)$ be a locally finite graph, possibly directed, and let $v \in V$. The ball of radius $\rho$ around $v$ in $G$ is
$$B_{\rho,G}(v) = \{ u \in V : d_G(v,u) \leq \rho\},$$
where $d_G(v,u)$ is the length of the shortest (directed) path from $v$ to $u$ in $G$.

For a rooted tree $T$ and a node $v$ in the tree, we consider all edges of $T$ as directed outward; the definition of $B_{\rho,T}(v)$ is then given by setting $G = T$ above and only considering directed paths along the tree.

The \emph{volume sequence} $V_G(v)$ of $G$ around $v$ is the sequence of non-negative integers
$$
|B_{0,G}(v)|, \ |B_{1,G}(v)|, \ |B_{2,G}(v)|, \ldots
$$ 
where we write $V_G(v, \rho) = |B_{\rho, G}(v)|$.
\end{definition}

Observe that for an undirected graph $G$, root $r$ and spanning tree $T$ of $G$ rooted at $r$ as generated by our process, it holds that $B_{\rho, G}(r) = B_{\rho, T}(r)$ for any ball of radius $\rho$. In particular, the volume sequence of $G$ and $T$ around $r$ are identical. Note also that $B_{\rho, T}(v) \subseteq B_{\rho, G}(v)$ with equality if and only if $v = r$. 

The second definition, of opinion bias, is closely related to the proportion of nodes holding the correct opinion among all opinionated (non-$\none$) nodes. For example, a bias of $1$ means that all nodes are correct; a bias of $0$ means that the aggregated opinion is uncorrelated with the truth.

\begin{definition}[Opinion bias]
\label{def:opinion_bias}
Consider any of the information spread processes,
with parameters $a$ and $b$ over a graph $G = (V, E)$ with root $r$. Let $f_t(v)$ denote the opinion of node $v$ at time $t$. The \emph{opinion bias} of the process at time $t$ is defined as
$$
\lambda_{G,r}(t) = \frac{|\{v \in V : f_t(v) = 1\} - \{v \in V : f_t(v) = -1\}|}{\{v \in V : f_t(v) \neq \none\}}.
$$
\end{definition}
Note that the opinion bias is equivalently the mean of non-$\none$ opinions across a single instance of the process. We stress that this \emph{not} an expectation over the randomness of the process, but rather a quantity associated with a single instance.

Our main theorem connects the expected value of the opinion bias at a certain time (over the randomness of the process) to the volume growth of the graph with respect to the root.

\begin{theorem}[Expected bias towards correct opinion]
\label{thm:unreliability_expectation}
Fix $a \in (0,1/2)$ and $b \in (0,1)$. The expected bias of the original spread process with parameters $(G, r, T, a, b)$ at time $t$ satisfies
\begin{equation*}
%\label{eq:expectation_unrel}
%\frac{1}{2} %\frac{1-e^{-t^{\Theta(1)}}}{2} \cdot
\left(1-e^{-t^{\Theta(1)}}\right) \cdot \frac{ V_G\left(r, bt - \Theta\left(t^{1-C}\right)\right)}{V_G(r, t)}
\leq
\mathbb{E}[\lambda_{G, r}(t)]
\leq 
\frac{ V_G\left(r, bt + \Theta\left(t^{1-C}\right)\right)}{V_G(r, t)} + e^{-t^{\Theta(1)}}
\end{equation*}
where $C > 0$ is an absolute constant. %and the $o(1)$ terms are of the form $e^{-t^{\Theta(1)}}$.
\end{theorem}
For convenience, when $G$ or $T$ are clear from context, we denote by $d(r, v)$ the distance from $r$ to $v$ in $G$ or $T$ (recall that these distances are equal as $T$ is a BFS tree of $G$ rooted at $r$).
\begin{proof}
%Let $T = (V, E_T)$ denote the communication tree spanning $G$ that is randomly picked by the process. 
As usual, we consider the decoupled process over $T$ with $\alpha = 2a$ and $\beta = b$ (see Lemma \ref{lem:equivalence_models}).
let $v \in V$ be any non-root node, and consider the branch from $r$ to $v$ in $T$.
By Lemmas \ref{lem:lb_truth_frontier} and \ref{lem:run_branch_path} and Theorem \ref{thm:ub_truth_frontier}, the truth frontier along the branch is with probability $1-e^{-t^{\Theta(1)}}$ in location $\min \{bt \pm O(t^{1-c}), d(r,v)\}$ for some absolute constant $c > 0$. 
Define for any $v \in B_{t,G}(r)$ (i.e., for any node with a non-$\none$ opinion) an indicator random variable $I_v = I_v(t)$ which equals one if the opinion of $v$ at time $t$ is $f_t(v) = 1$, and zero if $f_t(v) = -1$. We consider the different cases: when $d(r,v)$ is substantially smaller than the typical progress of the truth frontier, when $d(r,v)$ is much larger, and the in-between case.
\begin{itemize}
\item If $d(r,v) < bt - \Theta(t^{1-c})$, then the probability that the truth frontier has reached $v$ is very high, at $1 - e^{-t^{\Theta(1)}}$. In this case, 
$$\E[I_v(t)] = \Pr(I_v(t) = 1) \geq 1 \cdot \Pr\left(\text{truth frontier reached $v$}\right) =
1 - e^{-t^{\Theta(1)}}.$$
\item If $d(r, v) > bt + \Theta(t^{1-c})$, then the probability that the truth frontier reached $v$ is very small, at $\tau = e^{-t^{\Theta(1)}}$. When this event does not hold, the value of $v$ is uniformly random in $\{-1, 1\}$ (see Definition \ref{def:rooted_tree_decoupled}). Thus, 
$$
\E[I_v(t)] = 1 \cdot \tau + \frac{1}{2} \cdot (1-\tau) = \frac{1}{2} + e^{-t^{\Theta(1)}} \geq \frac{1}{2}.
$$
\item Otherwise, $d(r, v) \in bt \pm \Theta(t^{1-c})$. Here $1/2 \leq \E[I_v] \leq 1$, where the first inequality holds similarly to above: if the truth frontier has not reached a node $v$, then it is uniformly distributed in $\{-1, 1\}$. 
\end{itemize}
To complete the proof we may use the linearity of expectation, using the fact that
$$\E[\lambda_{G, r}(t)] = \frac{\sum_{v \in B_t(r)} \E[I_v(t)]}{2V_G(r, t)} - \frac{1}{2}.$$
It thus remains to give upper and lower bounds for $\sum_{v \in B_t(r)} \E[I_v(t)]$.

For the lower bound, note that there are $V_G(r, bt-\Theta(t^{1-c}))$ nodes satisfying the condition of the first bullet; they contribute $(1-e^{-t^{\Theta(1)}}) \cdot V_G(r, bt-\Theta(t^{1-c}))$ to the sum. All other nodes in $B_r(t)$ contribute to the sum at least $1/2$, and the proof follows.

For the upper bound, trivially every node satisfying the conditions of either the first or the third bullet contributes at most one to the numerator. There are up to $V_G(r, bt+\Theta(t^{1-c}))$ such nodes. Meanwhile, the contribution of every node satisfying the second condition is bounded by $\frac{1}{2} + e^{-t^{\Theta(1)}}$. This completes the proof.
\end{proof}

\paragraph{Polynomial versus super-polynomial growth.}
Ideally, one would hope that to have a positive bias which is as close to one as possible in all parts of the process, as this would mean that the ``average'' node is more likely than not to hold the correct opinion. But is indeed the case for graphs of interest? our results indicate that the answer is negative.

For the path $P_n$, our prior results show that the opinion bias at time $t$ is $b \pm o(1)$ both with high probability and in expectation. What can we say about other graphs? in particular, for what graphs (and at what times) is the opinion bias typically $o(1)$? when is it typically close $1$ (meaning the opinions are reliable at any time throughout the process)?

Theorem \ref{thm:unreliability_expectation} provides a rather clean characterization of the behavior of the opinion bias \emph{in expectation}. However, 
as we shall see later, the behavior in expectation does not give the full picture; there exist graphs in which the opinion bias may deviate substantially from the expectation.

In any case, we now address the above questions, focusing on the \emph{expected} opinion bias as a function of the graph structure.

\begin{remark}[Volume growth and finite graphs]
When talking about graphs with any sort of volume growth, say polynomial or exponential, it is convenient to think of the graph as an infinite or very large finite object. In finite graphs, the growth rate eventually converges to $1$, and is of order $1+o(1)$ already after a number of rounds that roughly equals the graph diameter divided by $b$ (this is true regardless of the choice of root $r$).

Thus, in finite graphs, the focus is more so on the earlier phases of the process, where very fast volume growth is possible (and in fact common for many theoretical graph models and real world graphs). In later parts of the process, the growth rate decreases, and consequently Theorem \ref{thm:unreliability_expectation}. This is unsurprising, as in later stages, typically all nodes converge to the correct opinion and so the bias converges to one.
\end{remark}

\paragraph{The story in the polynomial case.}
Consider any graph $G$ and root $r$ where $V_G(r, t) = \Theta(t^k)$ for a fixed $k > 0$, i.e., a polynomial growth. The discrete grid in $d$ dimensions is the most illustrative example (with $k=d$). Here, the volume growth satisfies
$$
V_G(r, bt \pm t^{1-\Theta(1)}) = \Theta(b^k) V_G(r, t),
$$
and so the expected opinion bias is $\Theta(b^k)$. In particular, if $V_G(r, t) = (c + o(1)) t^k$ for some absolute constant $c$ (which is the case for most graphs of interest in this regime, including the $d$-dimensional grid), then the expected opinion bias at time $t$ is $(1 \pm o(1)) b^k$. In all cases, the $o(1)$ term tends to zero as $t \to \infty$.

That is, the expected opinion bias for graphs with polynomial growth is bounded away from zero and one, and decays exponentially as a function of $k$.

\paragraph{The super-polynomial case.}
Many classes of graphs of interest, such as expanders and small world models, exhibit a super-polynomial growth. That is, these graphs have a volume growth of $V_G(r,t) = t^{\omega(1)}$, where the $\omega(1)$ goes to infinity as $t \to \infty$. In this case, 
$$
V_G(r, bt \pm t^{1-\Theta(1)}) = b^{\omega(1)} V_G(r, t),
$$
and so by Theorem \ref{thm:unreliability_expectation} the expected bias, of the form $b^{\omega(1)}$, tends to zero as $t \to \infty$. In other words, in graphs with super polynomial volume growth, the opinion of an ``average'' node at a certain time is nearly uncorrelated to the correct opinion.

\subsection{Concentration results and high probability proofs}
While Theorem \ref{thm:unreliability_expectation} holds in expectation for every graph $G$, there are examples of graphs where the opinion bias is not well-concentrated around its expected value. Consider for example a binary tree $T$ of depth $d$, containing $2^{d+1} - 1$ vertices; now add another $2^{d+2}$ vertices and connect them all to the same leaf $w$ in $T$, to form a star of this size. 
Observe the process at time $t=d+1$. The volume growth of the process is exponential, and by the results of the previous subsection, the expected opinion bias is $o(1)$. But is it also $o(1)$ with high probability? the answer is negative.

At time $t=d+1$, the additional star vertices all form their initial opinion, which with probability $1-a$ equals $f_d(w)$ and with probability $a$ is $-f_d(w)$. Therefore, the opinion bias among these nodes is either $2a-1$ or $1-2a$. From concentration results that we shall prove soon (Theorem~\ref{thm:concentration_super_poly} below), the other nodes at this time have with high probability opinion bias of $o(1)$. Therefore, the bias over all nodes is either $a-1/2+o(1)$ or $1/2-a+o(1)$, and in any case bounded away from zero assuming $a \neq 1/2$. 

\paragraph{A second moment approach.}
%The above example crucially relies on one of the nodes having a very large neighborhood. 

In order to prove concentration results for the opinion bias we employ a second moment approach on the decoupled process. The idea is that a pair of nodes $u,v$ that are not likely to appear in the same (decoupled) run at time $t$ have nearly uncorrelated opinions at this time, as their opinions are independent from each other conditioned on $u,v$ being in different runs. The following useful lemma provides sufficient conditions for $u,v$ to appear (with very high probability) in different runs at time $t$, thus allowing us to prove that correlations are weak between many pairs $u,v$.
\begin{lemma}
\label{lem:prob_same_run}
Fix $0 < \alpha, \beta < 1$ and $c > 0$. Consider the decoupled tree process $\dcpl$ with parameters $(G,r,T,\alpha, \beta)$ at some time $t > 0$, and let $u,v$ be a pair of nodes in $G$. Furthermore, let $w$ be the least common ancestor of $u,v$, denote by $D$ its distance from $r$ in $T$ and let $\delta = \max\{d_T(w, u), d_T(w, v)\}$. If 
$$\delta > \beta(t - D) + \Omega\left((t-D)^{0.5+c}\right) \ ,$$
then the probability that $u$ and $v$ are in the same run at time $t$ is at most $e^{-(t-D)^{\Theta(1)}}$.
\end{lemma}
\begin{proof}
Without loss of generality suppose that $\delta = d_T(w, u)$.
Let $k$ denote the index of the rumor frontier at $w$, when it initially formed an opinion at time $D$.
It suffices to prove that under the conditions of the lemma, as the $k$-frontier progresses along the branch containing $w$ and $u$, its probability to actually reach $u$ is very small. This however follows immediately from Lemma \ref{lem:progress_of_k_frontiers} applied with parameters $t^{(\ref{lem:progress_of_k_frontiers})} = t$ and $s^{(\ref{lem:progress_of_k_frontiers})} = D$, and $\delta^{(\ref{lem:progress_of_k_frontiers})} = C_{\alpha, \beta} (t-D)^{c-0.5}$ for a suitable constant $C_{\alpha, \beta} > 0$.
\end{proof}

With the above lemma in hand, we next show that the condition that most pairs of nodes are in different runs at a time $t$ is sufficient for concentration of the opinion bias $\lambda_{G,r}(t)$.
(Recall the notion of $\origin_t(v)$ from Definition \ref{def:rooted_tree_decoupled}, the index of the decoupled run of $v$.) 

\begin{lemma}
\label{lem:prob_origin_second_moment}
Consider the decoupled process $\dcpl$ with parameters $(G,r,T,\alpha, \beta)$ at time $t$. If
\begin{equation}
\label{eqn:origins}
\Pr_{(u,v) \in B_{t, G}(r)^2} \left( \origin_t(u) = \origin_t(v) \right) = o(1)
\end{equation}
where $u,v$ are uniformly and independently picked from the ball $B_t(r)$ of radius $t$ around $r$,
then the opinion bias $\lambda_{G,r}(t)$ is with probability $1-o(1)$ within $\pm o(1)$ of its expectation. Furthermore, the statement still holds if we replace all occurrences of the $o(1)$ term with suitable $e^{-t^{\Theta(1)}}$ terms.
\end{lemma}
\begin{proof}
The variance of $\lambda_{G,r}(t)$ is
$$
\text{Var}\left[\lambda_{G,r}(t)\right] = \frac{4}{(V_G(r, t))^2} \sum_{(u,v) \in B_t(r) \times B_t(r)}  \text{Cov}\left[I_t(u), I_t(v)\right],
$$
where by definition we have 
\begin{equation}
\label{eqn:cov}
\text{Cov}\left[I_t(u), I_t(v)\right] %= \E [I_t(u) I_t(v)] - \E [I_t(u)] \E [I_t(v)] 
= \left(\Pr(I_t(u) = 1 | I_t(v) = 1) - \Pr(I_t(u) = 1)\right) \cdot \Pr(I_t(v) = 1).
\end{equation}
Without loss of generality, by symmetry we may assume that $\origin_t(u) \geq \origin_t(v)$ for $u,v$ in the last expression. Conditioning on
$\origin_t(u) \neq \origin_t(v)$ (which by \eqref{eqn:origins} holds with probability $1-o(1)$) we know that $$\Pr(I_t(u) = 1 | I_t(v) = 1) = \Pr(I_t(u) = 1) = 1/2,$$
and so the covariance in \eqref{eqn:cov} is zero in this case. In any other case, the covariance is bounded by one. Therefore, the variance is bounded by a constant times the probability in \eqref{eqn:origins}, and so it is $o(1)$ (and $e^{-t^{\Theta(1)}}$ if the $o(1)$ is are of this form). The proof follows by Chebyshev inequality.
\end{proof}

We can now apply the above two lemmas to graph classes of interest. The first result is for graphs (and trees) with super-polynomial volume growth. The main idea is that a couple of nodes are likely to be in the same run at time $t$ only if they are both contained in some ball of radius up to $\approx bt$ in the tree $T$ generated in the process. When the volume growth is large enough, every such ball is substantially smaller than the radius-$t$ ball around the root (i.e., the collection of all non-$\none$ opinions).

\begin{theorem}[Concentration for super-polynomial volume growth]
\label{thm:concentration_super_poly}
Fix $a \in (0,1/2)$ and $b \in (0,1)$ and
consider the original process with parameters $(G,r,T,a,b)$. Let $t \in \N$ and suppose  
that the following two conditions hold for constants $0 < c \leq C$ and $0 < c' < 1-b$ and a function $w \colon \N \to \N$.
\begin{itemize}
    \item The function $\zeta(x) := \frac{w(x \cdot (b+c'))}{w(x)}$ converges to zero as $x \to \infty$.
    \item $V_T(v, x) \leq C \cdot x^{w(x)}$ for every $v \in B_{t,G}(r)$ and $x \leq t - d_G(r, v)$.
    \item $V_G(r, x) \geq c \cdot x^{w(x)}$ for every $x \leq t$. 
\end{itemize}
Then the opinion bias $\lambda_{G,r}(t)$ is $o(1)$ with probability $1-o(1)$, where the $o(1)$ terms tend to zero as $t \to \infty$. Furthermore, if we replace the $C \cdot x^{w(x)}$ and $c \cdot x^{w(x)}$ upper and lower bounds with expressions of the form $C e^{x^{\Omega(1)}}$ and $c e^{x^{\Omega(1)}}$ respectively (and ignore the first bullet), then the $o(1)$ terms are of the form $e^{-t^{\Theta(1)}}$. 
\end{theorem}
Note that the second bullet in the theorem can be trivially replaced with a similar condition about $G$ (thus making the statement independent of the tree $T$). The (stronger) condition is that $V_G(v,x) \leq C \cdot x^{\omega(x)}$ for every $v \in B_{t,G}(r)$ and $x \leq t - d_G(r,v)$.
\begin{proof}[Proof of Theorem \ref{thm:concentration_super_poly}]
Consider the process at time $t$. From Lemma \ref{lem:prob_origin_second_moment}, it suffices to prove \eqref{eqn:origins} for the graph $G$. Pick $u, v$ uniformly at random from $B_{t,G}(r)$. With probability $1-o(1)$, both $u$ and $v$ are of distance more than $(b+c')t$ from $r$. Condition on this event and set $c'' = c'/2$. From Lemma \ref{lem:prob_same_run}, the probability that $u,v$ are in the same run is $o(1)$ unless they both belong to some ball $B = B_{\rho,T}(w)$ where $\rho = (b+c'')t$. We show that the probability for this event to hold is $o(1)$ as well. Indeed, without loss of generality assume that $d_G(r,u) \geq d_G(r,v)$ and take $w$ as the unique node on the path from $r$ to $u$ of distance exactly $\rho$ from $u$. Then the last event holds only if $v \in B_{\rho, T}(w)$. The volume of this ball is bounded by 
$$
C \cdot \rho ^{\omega(\rho)} \leq
C \cdot t^{\omega\left((b + c'') \cdot t\right)} = o(t^{\omega(t)}),
$$
which proves the last claim.
The final statement of the theorem follows by replacing all the $o(1)$ terms in the proof with suitable $e^{-t^{\Theta(1)}}$ terms.
\end{proof}

The second result is for graphs with polynomial growth. The argument is quite different from the super-polynomial case. The main idea is that pairs of nodes $u,v$ that are very different in terms of their distance from the root -- specifically, the difference is much larger than $\approx \sqrt{t}$ -- are unlikely to appear in the same run at time $t$. The amount of nodes with $v$ that \emph{do not} satisfy this with respect to $u$ is sublinear in the total volume.
\begin{theorem}
[Concentration for polynomial volume growth]\label{thm:concentration_poly}
Fix $a \in (0,1/2)$ and $b \in (0,1)$ and
consider the original process with parameters $(G,r,T,a,b)$ at some time $t > 0$. Let $c > 0$ be an arbitrary small constant. Suppose that for every $s < s' \leq t$ with $s' - s \leq t^{0.5 + c}$, it holds that $V_G(r, s') - V_G(r,s) = o( V_G(r,t))$, where the $o(\cdot)$ term tends to zero as $t \to \infty$. Then the opinion bias at time $t$, $\lambda_{G,r}(t)$, is within $\pm o(1)$ of its expectation with probability $1-o(1)$.
\end{theorem}
The interpretation for finite graph is analogous to that discussed after Theorem \ref{thm:concentration_super_poly}.
\begin{proof}
From Lemma \ref{lem:prob_origin_second_moment} it suffices to prove that the probability of a uniform pair of nodes $u$ and $v$ to be in the same run at time $t$ is $o(1)$. From the statement of the lemma, the probability that $|d_G(r,u) - d_G(r,v)| \leq t^{0.5 +c}$ is $o(1)$. We now show that if this does not hold, then the probability of $u$ and $v$ to be in the same run is $o(1)$. 
Let $w$ be the least common ancestor of $u,v$ in the tree $T$ on which the process is run, and apply Lemma \ref{lem:progress_of_k_frontiers} with parameters $s=0$, $t$, and 
\end{proof}

A central example of a graph with polynomial volume growth is the $d$-dimensional (hyper-)grid. Here, a simple calculation shows that $V_G(r,t) \approx 2^d \binom{t+d}{d}$ for an infinite grid (or a finite grid when the distance of $r$ from any of its faces is at least $t$). It is straightforward to verify that the grid satisfies the conditions of Theorem \ref{thm:concentration_poly}. In combination with Theorem \ref{thm:unreliability_expectation}, we conclude that the opinion bias is $b^d + o(1)$ with probability $1-o(1)$, where both $o(1)$ terms tend to zero as $t \to \infty$.

\section{A slower variant of our model} \label{sec:slower} 
We briefly discuss here what happens in our model when instead of a choosing a BFS tree and correcting with respect to a fixed parent, each node, other than the root, error corrects with respect to a randomly chosen opinionated neighbor at each round. It is clear that this model will also converge to all nodes holding the correct opinion as this is the only absorbing state of the Markov chain defined by this process. 
However, we note that this model is exponentially slower to converge compared to our model. We first formally define the model by defining the update rule for all nodes different from the root. 

\begin{itemize}
\item Given the values of $f_{t-1}(v)$ for all $v \in V$, the update rule defining $f_{t}(\cdot)$ is as follows. 
\begin{itemize}
    \item \textbf{Before opinion formation.} If $f_{t-1}(v) = \none$ and the set $N_{t-1}(v) := \{u \in V: (u, v) \in E, f_{t-1}(u) \neq \none \}$ is empty, then $f_t(v) = \none$.
    \item \textbf{Initial opinion.} If $f_{t-1}(v) = \none$ and $N_{t-1}(v)$ is non-empty, we pick $p(v)$ uniformly at random from $N_{t-1}(v)$, and set
     \[
        f_t(v) = \begin{cases}
        f_{t-1}(p(v))  &\text{w.p. } 1-a,
        \\
        -f_{t-1}(p(v)) &\text{w.p. } a.
        \end{cases}
     \]
    \item \textbf{Opinion updates.} If $f_{t-1}(v) \neq \none$ then 
    we pick $p(v)$ uniformly at random from $N_{t-1}(v)$, and set
        \[
        f_t(v) = \begin{cases}
        f_{t-1}(p(v))  &\text{w.p. } b,
        \\
        f_{t-1}(v)) &\text{w.p. } 1-b.
        \end{cases}
        \]
\end{itemize}
\end{itemize} 

The model defined above is a simple variant of the voter model~\cite{HolleyLiggett:75}. Indeed using similar techniques to the one used there, one can prove that the convergence time of the process is determined by the volume of the graph instead of its diameter. 

\begin{lemma} \label{lem:voter_complete}
Consider the process above on the complete graph on $n$ vertices.
Then with probability $\Omega(a)$ the process converges in time $\Omega(n)$.
\end{lemma}
Of course, in our model on the complete graph the process converges in $O(\log n)$ steps. 

\begin{proof}
For simplicity we will assume that when a node corrects it may choose at random to correct by its own value. This corresponds to a $O(1/n)$ change in the value of $b$.

Let $X_t = \sum f_t(v) $. We are interested in the first time $T$ 
where $X_T = n$. 
Clearly $X_1 \sim 1 + Bin(n-1,1-a)$ which is concentrated around $(1-a)n$. 

Moverover, since when a non-root node updates its expected value is $X_t$ we get that: 
\[
E[X_{t+1} | X_t] \leq 1 + X_t.
\]
Therefore
\[
E[X_{a n / 4}] \leq a n / 4 + E[X_1] \leq (1-a/2)n
\]
and thus
\[
P[X_{a n / 4} \neq 1] \geq \frac{1}{n}E[n - X_{a n / 4}] \geq 0.4 a
\]
which completes the proof.
\end{proof}

We can prove a similar proof for other graphs such as the full binary tree.

\begin{lemma} \label{lem:voter_binary} 
Consider the process above on the full binary tree with $n$ vertices.
Then with probability $\Omega(a)$ the process converges in time $\Omega(n)$.
\end{lemma}
Again in our model on the complete graph the process converges in $O(\log n)$ steps.

\begin{proof}
The proof is similar. Let $h = \log_2(n)$ denote the height of the tree.
At time $h$ all nodes have value $0$ or $1$ and the nodes at the last level performed no error correction. 
The results of~\cite{KestenStigum:66,KestenStigum:67} imply that with high probability at time $h$ at most $3/4$ of the leaves take the 
correct value. 

Let $X(t) = \sum_{v} d(v) f_v(t)$. 
Then denoting the root by $\rho$
\begin{eqnarray*}
E[X(t+1) | f(t)] &=& \sum_v d(v) E[f_v(t+1) | f(t)] = 
d(\rho) + \sum_{v \neq \rho} d(v) (1-b) f_v(t) + b \sum_{w \sim v} f_w(t)
\\ &\leq& (1+b) d(\rho) +  \sum_v d(v) f_v(t) = X(t) + (1+b) d (\rho) 
\end{eqnarray*}
This implies that $X(t) - 2 d(\rho) t$ is a super-martingale (goes down in expectation) and the proof concludes as in Lemma~\ref{lem:voter_complete}. 
\end{proof}

We note that lemmas~\ref{lem:voter_complete} and~\ref{lem:voter_binary} 
can be strengthened to hold with probability $1-o(1)$ using second moment or a stronger form of concentration. 

\section{Information Spread with Online Tree Selection} 

\label{sec:c_less_1}
In this section we briefly discuss an alternate model where the tree $T$ is not specified in advance, but is a byproduct of the information spread process itself. A natural way of growing such a tree is to let
an unopinionated node $v$ look at all its neighbors in $G$, and among the subset of these that have an opinion at time $t$, pick one at random to be its parent. If none of the neighbors of $v$ has an opinion at time $t$, it does not pick a parent at time $t$. Once a node picks a parent, the choice is not changed in future steps. This leads to a natural way to pick $T$ randomly from $G$ and we refer to this as the {\em online} version.

As described above the tree chosen is a (random) BFS tree rooted at $r$. All results mentioned in previous sections continue to hold, as long as $T$ remains a BFS tree. However this equivalence fails to hold if we do not require that unopinionated nodes must form an opinion if some neighbor has an opinion. In this section we consider this model, where we have an 
an additional parameter $c$, so that any such node $v$ only forms an initial opinion with probability $c$, and otherwise remains unopinionated.
When $c=1$ we get the model from the previous paragraph which is close in spirit to our model from Definition \ref{def:model_general_Graphs}, while for $c < 1$ we get a more general model. (We note thus that to get new effects we really need {\em both} random growth of $T$ and $c<1$.)

We describe the model below and discuss it in the sequel.

%\madhu{end changes}
%In the offline case, where there is a BFS tree $T$ given to us as a parameter (for example, as some infrastructure formed in previous communications), many of the results in the paper apply in a straightforward manner to the generalized setting, at least in the regime $b < c$. The main difference is that the rumor frontier progresses at rate at least $c$, instead of a deterministic rate of $1$. The \emph{online} case is however much more interesting. Unlike the special case of $c=1$, where $T$ is always a BFS tree (both the offline and the online models), here the structure of the tree (formed online as part of the process) may be more complicated and is much less well understood. Thus, in what follows we focus our discussion on the online variant. \omri{Perhaps need to say that it is also possible to have the tree formed first with probability $c$ and then have the communication afterwards, so still there is some sensible offline analogue.}

\begin{definition}[Information spread with error correction: Online tree selection]
\label{def:model_c_less_1}
Fix three parameters $a,b,c \in [0,1]$.
Let $G = (V,E)$ be an undirected graph and let $r \in V$ denote some root vertex. Consider the following process with parameters $(G, r, a, b, c)$, proceeding in rounds. 

\begin{itemize}
\item For every round $t \geq 0$, each vertex $v \in V$ holds a label $f_t(v) \in \{+1, -1, \none\}$. As before, $f_0(r) = +1$ and $f_0(v) = \none$ for $v \neq r$.
\item For every $t \geq 0$ define a \emph{parent} function $p_t \colon V \to V \cup \{\none\}$, where $p_0(r) = r$ and $p_0(v) = \none$ for $v \neq r$. Parents never change once assigned: if $p_{t-1}(v) \neq \none$ then $p_{t}(v) = p_{t-1}(v)$.
\item Given the values of $f_{t-1}(v)$ and $p_{t-1}(v)$ for all nodes $v \in V$, the update rule for $f_t(\cdot)$ and $p_t(\cdot)$ is as follows. (The first and third bullets are essentially as in Definition \ref{def:model_general_Graphs}.)
\begin{itemize}
    \item \textbf{Before opinion formation.} If $f_{t-1}(v) = \none$ and the set $N_{t-1}(v) := \{u \in V: (u, v) \in E, f_{t-1}(u) \neq \none \}$ is empty, then $f_t(v) = \none$ and $p_t(v) = \none$. 
\iffalse{ 
    \item \textbf{Picking parent and initial opinion.} If $f_{t-1}(v) = \none$ and $N_{t-1}(v)$ is non-empty, 
    we act as follows. With probability $c$, we pick $p_t(v)$ arbitrarily from $N_{t-1}(v)$ and set
     \[
     f_t(v) = \begin{cases}
        f_{t-1}(p_t(v))  &\text{w.p. } 1-a,
        \\
        -f_{t-1}(p_t(v)) &\text{w.p. } a.
        \end{cases}
     \]
     Otherwise (with probability $1-c$), $v$ remains without unopinionated: $f_t(v) = p_t(v) = \none$.
}\fi 

    %\madhu{minor edits in this definition. $p(v)$ was missing subscript in several places. Also change "any" to "every"}
    %\madhu{Changing next item. I think this is a bit clearer, but if we prefer previous text we can switch easily - it is iffalse'd above.}
    %\omri{Not exactly the same because this is online and 2.1 is offline.}
    %\madhu{Is the comment above in response to mine? If the math below is incorrect we should switch back. But if it is the same process, then the below feels clearer. ``Pick X w.p. c and then Y with further probabilities'' can be confusing.}
    \item \textbf{Picking parent and initial opinion.} If $f_{t-1}(v) = \none$ and $N_{t-1}(v)$ is non-empty, 
    we act as follows. Pick $w$ arbitrarily from $N_{t-1}(v)$ and set
     \[
     f_t(v) = \begin{cases}
        f_{t-1}(w)  &\text{w.p. } c\cdot (1-a),
        \\
        -f_{t-1}(w) &\text{w.p. } c \cdot a, \\
        \none & \text{w.p. } 1-c. \\
        \end{cases}
     \]
    In the first two cases (i.e., when $f_t(v) \ne \none$), we set $p_t(v) = w$ and in the final case we set $p_t(v) = \none$.
    \item \textbf{Opinion updates.} If $f_{t-1}(v) \neq \none$ then
        \[
        f_t(v) = \begin{cases}
        f_{t-1}(p_t(v))  &\text{w.p. } b,
        \\
        f_{t-1}(v) &\text{w.p. } 1-b.
        \end{cases}
        \]
\end{itemize}
\end{itemize}

%\omri{We decided to first state the model in full generality, including the parameter $c$. Immediately after we will say that we consider $c=1$ until further notice, as it is a special case simpler to analyze. In one of the later sections we can discuss the case where $c \neq 1$ in more detail.}

Formally, the generalized spread model is given by the random variables $\{p_t(v)\}_{t \in \N, v \in V} \cup \{f_t(v)\}_{t \in \N, v \in V}$, with $p_t(v) \in V \cup \{\none\}$ and $f_t(v) \in \{-1,\none,+1\}$ for every $t,v$ generated as above.
\end{definition}

Which results carry over from our analysis of the offline case (which, as mentioned, also applies for the case $c=1$)? The speed of the rumor frontier now is $c$, while the speed of the truth frontier is $b$ as before, so the regime where one can hope to obtain similar results is when $b < c$.

\paragraph{A decoupled generalized process} As in the original process, the fact that runs can merge and therefore change their size rapidly seems difficult to analyze and it is better to work with a decoupled model. It is not hard to prove that a decoupled model similar to Definition \ref{def:rooted_tree_decoupled} is analogous to the generalized model in Definition \ref{def:model_c_less_1}. For brevity, here we only quickly mention the main modifications required to Definition \ref{def:rooted_tree_decoupled} in this case. 
\begin{itemize}
    \item The model receives three parameters $\alpha, \beta, \gamma \in [0,1]$. It also receives $G$ and the root $r$ as parameters (but no tree $T$).
    \item Picking a parent is done exactly as in Definition \ref{def:model_c_less_1}, with parameter $c$ replaced by $\gamma$. Specifically, all nodes except for $r$ start with no parent, and with probability $\gamma$ choose and start listening to a parent (which is fixed for the rest of the process) in any round where they are still unopinionated but have an opinionated neighbor.
    \item The update rule is the same as in Definition \ref{def:rooted_tree_decoupled}, except that the second bullet is carried only for unopinionated nodes that have just been linked to a parent in the current round.
\end{itemize}
As in Lemma \ref{lem:equivalence_models}, it is not hard to prove that this model is equivalent to that of Definition \ref{def:model_c_less_1} with the parameter setting $\alpha = 2a, \beta = b, \gamma = c$. As before, we are interested in the setting where $\alpha, \beta$ are constants bounded away from zero and one (which is equivalent to $a \in (0,1/2)$ and $b \in (0,1)$).

\paragraph{Over the path} Suppose first that the graph $G$ is the rooted path $\cP_n$ (where $n$ is finite or infinite). At any given round, the probability of the decoupled rumor frontier to progress by one is $\gamma$, while the probability of any intermediate frontier to progress is $\beta$. It is standard to show, then, that the rumor frontier at round $t$ is with probability $1-e^{-t^{\Theta(1)}}$ at location $\gamma t \pm O(t^{0.6})$. On the other hand, as soon as a certain frontier in the decoupled process becomes intermediate (i.e., separated from the rumor frontier), its behavior becomes as in the original process. In particular, the tight bounds on the truth frontier (Lemma \ref{lem:lb_truth_frontier}, Theorem \ref{thm:ub_truth_frontier}, and Corollary \ref{coro:truth_frontier}) still apply.

\paragraph{Convergence of individual opinion}
We shift our focus to the setting of a general graph $G$, and first mention which results carry over easily from the case $\gamma = 1$. starting with those results concerning when a node ``knows'' that it has the correct opinion (Theorem \ref{thm:changes_of_opinion_individual} and Corollary \ref{coro:when_to_stop}). Since these results only depend on what happens \emph{after} a node $v$ forms its initial opinion (at which time the path from $r$ to $v$ is fixed, and the progress of frontiers towards $v$ is always at rate $\beta$), they hold word for word in the general setting.

\paragraph{Tree structure in generalized model} Matters become more complicated when trying to prove results that crucially rely on the structure of the tree $T$ formed along the process. For example, write $d = d_G(r,v)$. In the path $\cP_n$, the node $v$ typically forms an initial opinion at round $\approx d / \gamma$ and the truth frontier arrives in $v$ at time $\approx d / \beta$. Do the same results still hold for general $G$? 

The answer turns out negative: While the above reasoning gives an upper bound for any $G$, it is not tight in general. Suppose for example that there are $C^d$ disjoint paths of length exactly $d$ from $r$ to $v$ in $G$, where $C = C(\alpha, \beta, \gamma)$ is a large enough constant. Consider the rumor frontier along each of them. For any such path $\cP$ in which the predecessor of $v$ is denoted $u_\cP$, with probability $c^{d-1}$ both the rumor and truth frontiers will reach $u_\cP$ in exactly $d-1$ rounds. For $C$ large enough, with probability $0.99$ this event will hold for at least one of the paths $\cP$. Conditioning on this happening, $v$ will typically receive an initial opinion at round $d+O(1)$.

The above example demonstrates that progress of frontiers in the general case depends not only on the distances of nodes from $r$ in $G$, but also on the profile of paths leading up to each of the nodes in $G$. We thus leave the following as an open question.
\begin{question}
What can we say about the typical structure of the communication tree $T$ formed in the generalized process, as a function of $G$? What parameters of the graph $G$ may lead the opinion bias to improve (increase) in the generalized process as compared to the case $c=1$?
\end{question}

\bibliographystyle{plain}
\bibliography{references,all,my}

\begin{thebibliography}{10}

\bibitem{Banerjee:92}
Abhijit~V. Banerjee.
\newblock A simple model of herd behavior.
\newblock {\em The Quarterly Journal of Economics}, 107(3):797--817, 1992.

\bibitem{BlRuZa:95}
P.~M. Bleher, J.~Ruiz, and V.~A. Zagrebnov.
\newblock On the purity of the limiting {G}ibbs state for the {I}sing model on
  the {B}ethe lattice.
\newblock {\em J. Statist. Phys.}, 79(1-2):473--482, 1995.

\bibitem{BCMR:06}
C.~Borgs, J.~Chayes, E.~Mossel, and S.~Roch.
\newblock The {K}esten-{S}tigum reconstruction bound is tight for roughly
  symmetric binary channels.
\newblock In {\em Proceedings of IEEE FOCS 2006}, pages 518--530, 2006.

\bibitem{EvKePeSc:00}
W.~S. Evans, C.~Kenyon, Yuval Y.~Peres, and L.~J. Schulman.
\newblock Broadcasting on trees and the {I}sing model.
\newblock {\em Ann. Appl. Probab.}, 10(2):410--433, 2000.

\bibitem{EvansSchulman:99}
W.~S. Evans and L.~J. Schulman.
\newblock Signal propagation and noisy circuits.
\newblock {\em IEEE Trans. Inform. Theory}, 45(7):2367--2373, 1999.

\bibitem{HJMR:19b}
Jan H\k{a}z\l{}a, Ali Jadbabaie, Elchanan Mossel, and M.~Amin Rahimian.
\newblock Reasoning in {B}ayesian opinion exchange networks is
  {$\mathsf{PSPACE}$}-hard.
\newblock In {\em Proceedings of the Thirty-Second Conference on Learning
  Theory (COLT)}, pages 1614--1648, 2019.

\bibitem{HolleyLiggett:75}
A.~Holley and Thomas~M. Liggett.
\newblock Ergodic theorems for weakly interacting infinite systems and the
  voter model.
\newblock {\em Ann. Probab.}, 3:643--663, 1975.

\bibitem{Ioffe:96b}
D.~Ioffe.
\newblock Extremality of the disordered state for the {I}sing model on general
  trees.
\newblock In {\em Trees (Versailles, 1995)}, volume~40 of {\em Progr. Probab.},
  pages 3--14. Birkh\"auser, Basel, 1996.

\bibitem{Ioffe:96a}
D.~Ioffe.
\newblock On the extremality of the disordered state for the {I}sing model on
  the {B}ethe lattice.
\newblock {\em Lett. Math. Phys.}, 37(2):137--143, 1996.

\bibitem{Jackson:08}
M.~O. Jackson.
\newblock {\em Social and Economic Networks}.
\newblock Princeton University Press, 2008.

\bibitem{KestenStigum:66}
H.~Kesten and B.~P. Stigum.
\newblock Additional limit theorems for indecomposable multidimensional
  {G}alton-{W}atson processes.
\newblock {\em Ann. Math. Statist.}, 37:1463--1481, 1966.

\bibitem{KestenStigum:67}
H.~Kesten and B.~P. Stigum.
\newblock Limit theorems for decomposable multi-dimensional {G}alton-{W}atson
  processes.
\newblock {\em J. Math. Anal. Appl.}, 17:309--338, 1967.

\bibitem{Liggett:85}
Thomas~M. Liggett.
\newblock {\em Interacting particle systems}, volume 276 of {\em Grundlehren
  der Mathematischen Wissenschaften [Fundamental Principles of Mathematical
  Sciences]}.
\newblock Springer-Verlag, New York, 1985.

\bibitem{MaMoPo:20}
Anuran Makur, Elchanan Mossel, and Yury Polyanskiy.
\newblock Broadcasting on random directed acyclic graphs.
\newblock {\em IEEE Information Theory}, 66(2):780--812, 2020.

\bibitem{MontanariSaberi:10}
Andrea Montanari and Amin Saberi.
\newblock The spread of innovations in social networks.
\newblock {\em Proceedings of the National Academy of Sciences},
  107(47):20196--20201, 2010.

\bibitem{MosselTamuz:17}
E.~Mossel and O.~Tamuz.
\newblock Opinion exchange dynamics.
\newblock {\em Probability Surveys}, 14:155--204, 2017.

\bibitem{Ritter:81}
Grant~A Ritter et~al.
\newblock Growth of random walks conditioned to stay positive.
\newblock {\em Annals of Probability}, 9(4):699--704, 1981.

\bibitem{ScheufeleKrause:19}
Dietram~A Scheufele and Nicole~M Krause.
\newblock Science audiences, misinformation, and fake news.
\newblock {\em Proceedings of the National Academy of Sciences},
  116(16):7662--7669, 2019.

\bibitem{SmithSorensen:00}
L.~Smith and P.~Sorensen.
\newblock Pathological outcomes of observational learning.
\newblock {\em Econometrica}, 68(2):371--398, 2000.

\bibitem{BiHiWe:98}
David~Hirshleifer Sushil~Bikhchandani and Ivo Welch.
\newblock Learning from the behavior of others: Conformity, fads, and
  informational cascades.
\newblock {\em The Journal of Economic Perspectives}, 12(3):151--170, 1998.

\bibitem{vonNeumann:56}
J.~von Neumann.
\newblock Probabilistic logics and the synthesis of reliable organisms from
  unreliable components.
\newblock In {\em Automata studies}, Annals of mathematics studies, no. 34,
  pages 43--98. Princeton University Press, Princeton, N. J., 1956.

\bibitem{Wylie:19}
Christopher Wylie.
\newblock {\em Mindf$*$ck: Cambridge Analytica and the Plot to Break America}.
\newblock Random House, New York, 2019.

\end{thebibliography}

\end{document}